\documentclass[a4paper,12pt]{article}
\usepackage{amsmath,amsthm,amssymb,bbm}

\usepackage[]{inputenc}
\usepackage[T1]{fontenc}
\usepackage[english]{babel}
\usepackage[]{csquotes}
\usepackage[]{graphicx}
\usepackage{subfigure}
\usepackage[]{cite}
\usepackage{amssymb}
\usepackage{bbm}
\usepackage{braket}

\newtheorem{theorem}{Theorem}

\newtheorem{lemma}{Lemma}
\newtheorem{definition}{Definition}
\newtheorem{proposition}{Proposition}

\newtheorem{remark}{Remark}

\newcommand{\Tr}{{{\rm Tr}}}

\title{Dissipative Dynamics of Quantum Fluctuations}

\author{F. Benatti$^{1,2}$, F. Carollo$^{1,2}$, R. Floreanini$^2$\\
\small ${}^1$Dipartimento di Fisica, Universit\`a di Trieste, Trieste, 34151 Italy\\
\small ${}^2$Istituto Nazionale di Fisica Nucleare, Sezione di Trieste, 34151 Trieste, Italy}

\date{\null}

\begin{document}

\maketitle

%\author{Fabio Benatti}
%\affiliation{Dipartimento di Fisica, Universit\`a di Trieste, 34151
%Trieste, Italy}
%\affiliation{Istituto Nazionale di Fisica Nucleare, Sezione di Trieste, 34151
%Trieste, Italy} 

%\author{Federico Carollo}
%\affiliation{Dipartimento di Fisica, Universit\`a di Trieste, 34151
%Trieste, Italy}

%\author{Roberto Floreanini}
%\affiliation{Istituto Nazionale di Fisica Nucleare, Sezione di Trieste, 34151
%Trieste, Italy} 

\begin{abstract}
One way to look for complex behaviours in many-body quantum systems is to let the number $N$ of degrees of freedom become large and focus upon collective observables.
Mean-field quantities scaling as $1/N$ tend to commute, whence complexity at the quantum level can only be inherited from complexity at the classical level. 
Instead, fluctuations of microscopic observables scale as $1/\sqrt{N}$ and exhibit  collective Bosonic features, typical of a mesoscopic regime half-way between the quantum one at the microscopic level and the classical one at the level of macroscopic averages.
Here, we consider the mesoscopic behaviour emerging from an infinite quantum spin chain undergoing a microscopic dissipative, irreversible dynamics and from global states without long-range correlations and invariant under lattice translations and dynamics.
We show that, from the fluctuations of one site spin observables whose linear span is mapped into itself by the dynamics, there emerge bosonic operators obeying
a mesoscopic dissipative dynamics mapping Gaussian states into Gaussian states. Instead of just depleting quantum correlations because of decoherence effects, these maps can generate entanglement at the collective, mesoscopic level, a phenomenon with no classical analogue that embodies a peculiar complex behaviour at the interface between micro and macro regimes.
\end{abstract}

\section{Introduction}
In many-body quantum systems, when the number $N$ of
constituents becomes very large, the accessible observables 
are collective ones.
For these ``macroscopic'' observables, one usually expects that quantum effects 
fade away, even more so when
the many-body system is in contact with an external environment.
This is surely the case for ``mean field'' observables, that are  averages 
of single particle microscopic operators: they scale as $1/N$
and behave as classical commuting observables when $N$ becomes large.
The only complex behaviour they can show is inherited from a possibly complex classical limiting behaviour: this is the scenario of the theory of quantum chaos where one studies the footprints left by the quantisation of chaotic classical systems either in the appearance of typical logarithmic time-scales or energy spectrum distributions\cite{chaos}.

Instead, collective observables scaling as
$1/\sqrt{N}$ \cite{Goderis,Verbeure,Matsui} retain 
quantum properties as $N$ increases: they have been called
``fluctuation operators'' since they account for global deviations from mean values as classical fluctuations do. Indeed, for them a quantum central limit theorem has been 
proved with respect to states without spatial long-range correlations.
These fluctuation operators form an algebra that, independently from the nature of the microscopic many-body system, turns out to be 
non-commutative and always of bosonic type, thus showing
a quantum behaviour. Being half-way between microscopic
observables (as for instance the individual spin operators 
in generic spin systems) and truly macroscopic ones 
({\it e.g.} the corresponding mean magnetization),
the fluctuation operators have been named ``mesoscopic'':
they are the place where to look for truly quantum signals
in the dynamics of ``large'' systems.

Unlike in the mean-field scenario where complexity can only be inherited from the classical level of description, at the mesoscopic quantum level it can manifest itself as an emergent phenomenon through the presence of those typical quantum features as entanglement that have no classical counterpart. While it is common in quantum systems with few degrees of freedom,
such a presence is in general quite unexpected in large quantum systems where the emergent mesoscopic dynamics is likely to be marred by strong decoherence effects which spoil quantumness, especially if the 
microscopic dynamics is itself dissipative as in the following.

The emergent dynamics of fluctuations has been studied in the case of reversible, that is unitary, microscopic dynamics \cite{Verbeure,Matsui}; instead, very little is known for open many-body systems, {\it i.e.} for systems immersed in an external bath.
This is the most common situation encountered in actual experiments, typically
involving cold atoms, optomechanical or spin-like systems \cite{Aspelmeyer,Rogers}, that can never be thought of as completely isolated from their thermal surroundings.
Actually, the repeated claim of having detected ``macroscopic'' entanglement
in those experiments \cite{Jost,Krauter} poses a serious challenge in trying to 
interpret theoretically those results \cite{Narnhofer}.

Motivated by these experimental findings, it has been shown that, in concrete models, quantum behaviour can indeed be present at the mesoscopic level in open many-body
systems provided suitable fluctuation operators are considered. Even more strikingly, mesoscopic entanglement can be induced by the presence of an external environment \cite{BCF}. These evidences could be obtained for the case of two quantum spin chains undergoing a Lindblad type dissipative dynamics that statistically couples their spins and preserves the global microscopic state.
In such model, a Weyl algebra of quantum fluctuations of single site spin operators could be  constructed and a mesoscopic semigroup of Gaussian dynamical maps on the states over the Weyl algebra could be derived from the given dissipative microscopic time-evolution.
However, some assumptions have been made, namely
\begin{itemize}
\item
the microscopic state on the quantum spin chains was considered to be a \textit{KMS} thermal factor state, namely a tensor product of a same Gibbs density matrix at each site;
\item 
the generator of the microscopic dissipative dynamics had a specific form.
\end{itemize}

In the following, we will instead consider a quantum spin chain (of which the chain considered in \cite{BCF} is a special case) endowed with a generic time-invariant clustering state, without any request of being either a tensor product or a thermal state, and with a microscopic dissipative dynamics that only obeys certain consistency constraints with respect to the chosen quantum fluctuations. In this setting, we show that, generically, the microscopic dissipative dynamics of an open quantum spin chain gives rise to a mesoscopic semigroup of completely positive, unital maps that transform Weyl operators into Weyl operators so that the dual maps acting on the states over the Weyl algebra maps Gaussian states into Gaussian states.
These results are relevant for actual applications, in particular in the description of the behaviour of ultra-cold gases trapped in optical lattices that usually involve large numbers of atoms distributed over on-site confining potentials and show collective, coherent quantum behaviours.

\section{Algebra of Quantum Fluctuations}

In this section we shall briefly review the construction of the algebra of fluctuations for a quantum spin chain, namely for a one-dimensional lattice with a finite ($p$-level) spin, described by the matrix algebra $\mathcal{A}^{(k)}=M_{p}(\mathbb{C})$ at each site $k$, where $M_p(\mathbb{C})$ denotes the algebra of $p\times p$ complex matrices. 

Technically speaking, the infinite spin chain is described by the ($C^*$) algebra $\mathcal{A}$, called "quasi-local", that arises from the norm-closure of the union of all local algebras supported by finite subsets of the lattice \cite{Bratteli}:
\begin{equation}
\label{qla0}
\mathcal{A}=\overline{\bigcup_{N\geq 0}\mathcal{A}_{[-N,N]}}^{\|\cdot\|}\ .
\end{equation}
The physical meaning of this construction is straightforward: all operators $A\in\mathcal{A}$ can be approximated in norm as well as one wishes by means of local operators $A_{[-N,N]}$ belonging to finitely supported matrix algebras $\mathcal{A}_{[-N,N]}=\bigotimes_{k=-N}^N \mathcal{A}^{(k)}$.
 
Given a single spin operator $a\in M_{p}(\mathbb{C})$, it is embedded into $\mathcal{A}$ as an element of the $k$-site spin algebra as follows:
\begin{equation}
\label{qla1}
M_p(\mathbb{C})\ni a\mapsto a^{(k)}=1_{-\infty,k-1]}\otimes a\otimes 1_{[k+1,\infty}\ ,
\end{equation}
where $1_{-\infty,k-1]}$, respectively $1_{[k+1,\infty}$ denote tensor products of infinitely many identity matrices up to site $k-1$, respectively from site $k+1$.
Then, one can endow the quasi-local algebra $\mathcal{A}$ with the translation 
automorphism $\tau:\mathcal{A}\to\mathcal{A}$, whose action is to shift the operator from site $k$ to site $k+1$: $\tau\left(a^{(k)}\right)=a^{(k+1)}$.

States on $\mathcal{A}$ are all linear functionals 
$\omega:\mathcal{A}\mapsto\mathbb{C}$ that are positive, $\omega(A^\dag A)\geq 0$ for all $A\in\mathcal{A}$, and normalized, $\omega(1)=1$: they indeed correspond to generic expectations on $\mathcal{A}$ that assign to operators $A\in\mathcal{A}$ their mean values $\omega(A)$.
A translation invariant state $\omega$, is a state such that 
\begin{equation}
\label{timeinv0}
\omega(\tau(A))=\omega(A)\qquad\forall\, A\in\mathcal{A}\ .
\end{equation}
A particular class of translation invariant states are the clustering states, namely states with vanishing spatial correlations among sufficiently far apart operators:
\begin{equation}
\label{qla3}
\lim_{|k|\to\infty}\omega(\tau^k(A)B)=\omega\left(A\right)\omega(B)\qquad \forall A,B\in\mathcal{A}\ .
\end{equation}

\subsection{Quantum Fluctuations}

In order to construct the algebra of quantum fluctuations for the quantum spin chain $\mathcal{A}$, one selects a translation invariant, clustering state $\omega$ and a specific set of hermitian spin operators.
In the following, as in \cite{BCF}, we shall focus upon a set $\chi$ consisting of 
single-site operators $\chi=\left\{x_1,x_2,\dots,x_d\right\}\subseteq M_{p}(\mathbb{C})$. 

\begin{remark}
\label{rem-1}
The set $\chi$ of observables of which one considers the fluctuations need not be single-particle operators: they can be supported by a larger number of lattice sites.
They need not either be a generating set for the local algebra they belong to: for instance, in the case of spins $1/2$, $\chi$ may or may not consist of all the Pauli matrices plus the identity matrix. Both these choices are ultimately dictated by the mesoscopic physics one is interested in. 
\end{remark}

Each $x_i\in\chi$ is characterised by a local fluctuation operator
\begin{equation}
\label{fqf0}
F_N(x_i)=\frac{1}{\sqrt{N_T}}\sum_{k=-N}^{N}\left(x_i^{(k)}-\omega(x_i)\right)\ ,\qquad
N_T=2N+1\ .
\end{equation}
Notice that 
\begin{equation}
\label{lqf0}
\omega\Big(F_N(x_i)\Big)=0\ .
\end{equation}
The local fluctuation operators provide a quantum version of the classical fluctuations of a family of identically distributed stochastic variables, where now the latter are replaced by non-commuting spins and their stationary joint probability distributions by the expectations with respect to the invariant state $\omega$.
For spin chain quantum fluctuations, the reformulation of the classical central limit theorem \cite{Verbeure} can be briefly summarized as follows.

\begin{definition}
The system $\left(\omega,\chi\right)$ with $\omega$ translation invariant and clustering state, is said to have normal quantum fluctuations if, $\forall x_i,x_j\in \chi$,
\begin{eqnarray*}
&&
\sum_{k\in\mathbb{Z}}\left|\omega\left(x_i\tau^k\left(x_j\right)\right)-\omega(x_i)\omega(x_j)\right|<\infty\\
&&
\lim_{N\to+\infty}\omega\left(F^2_N(x_i)\right)=:\Sigma_\omega^{ii}\ ,\qquad
\lim_{N\to+\infty}\omega\left({\rm e}^{i\alpha F_N(x_i)}\right)=\exp\left(-\frac{\alpha^2}{2}\Sigma_\omega^{ii}\right)\qquad \forall \alpha\in\mathbb{R}\ .
\end{eqnarray*}
\label{NQF}
\end{definition}

Consider the commutator of two fluctuation operators: because operators at different sites commute one gets
$$
\Big[F_N(x_i),F_N(x_j)\Big]=
\frac{1}{N_T}\sum_{k=-N}^Nz_{ij}^{(k)}=:<z_{ij}>_N\ ,\quad
z^{(k)}_{ij}=\left[x^{(k)}_i\,,\,x^{(k)}_j\right]\ ,
$$
namely, the average $<z_{ij}>_N$ of a same single-site spin operator, 
$z_{ij}=[x_i,x_j]$, over the $N_T=2N+1$ sites between $-N$ and $N$, a typical mean-field observable. Furthermore, since $\omega$ is translation invariant, 
$\omega\left(<z_{ij}>_N\right)=\omega\left(z_{ij}\right)$, and clustering, one deduces that
$$
\lim_{N\to+\infty}\omega(A^\dag\,<z_{ij}>_N\,B)=\omega(A^\dag\,B)\,\omega(z_{ij})
$$
for all $A,B\in\mathcal{A}$.
Namely, commutators of local fluctuations tend in a proper (weak) sense to multiples of the identity, a behaviour which is  completely different from that of the commutators of 
mean-field observables. Indeed, $\displaystyle <x_i>_N=\frac{1}{N_T}\sum_{k=-N}^Nx_i^{(k)}$ and
$$
\left[<x_i>_N\,,\,<x_j>_N\right]=\frac{1}{N_T^2}\sum_{k=-N}^Nz_{ij}^{(k)}
$$  
vanishes as $1/N_T$ for $\|z^{(k)}_{ij}\|=\|z_{ij}\|$ is finite.

While mean-field observables behave as commuting classical, macroscopic observables in the large $N$ limit, in the same limit local fluctuations behave as bosonic degrees of freedom. In fact, commutators of local fluctuations not only do not vanish, but also  provide a symplectic form, $\sigma_\omega$, on the linear span of the elements of the chosen set $\chi$. Indeed, the quantities
\begin{equation}
\label{qfa0a}
\sigma_\omega^{k\ell}=-i\lim_{N\to+\infty}\omega\left(\left[F_N(x_k),F_N(x_\ell)\right]\right)
\end{equation}
are real and satisfy $\sigma_\omega^{k\ell}=-\sigma_\omega^{\ell k}$. 

By means of the symplectic form, one can introduce Weyl-like operators indexed by the elements $x_i\in\chi$ satisfying $W(x_k)^\dagger=W(-x_k)$ and the Canonical Commutation Relations
\begin{equation}
\label{qfa0b}
W(x_k)W(x_\ell)=W(x_k+x_\ell)\, \exp\Big(-\frac{i}{2}\sigma_\omega^{k\ell}\Big)\ .
\end{equation}
The Weyl algebra $\mathcal{W}(\chi,\sigma_\omega)$ is the algebra generated by the linear span of generic Weyl operators, indexed by real vectors $r\in\mathbb{R}^d$,
\begin{equation}
\label{qfa4}
W(r):=W\Big(\sum_{j=1}^dr_j\,x_j\Big)\ ,\quad W(r_1)\,W(r_2)=W(r_1+r_2)\, {\rm e}^{-i\sigma_\omega(r_1,r_2)/2}\ ,
\end{equation}
where $\displaystyle \sigma_\omega(r_1,r_2)=\sum_{k,\ell=1}^d r_{1k}r_{2\ell}
\,\sigma_\omega^{k\ell}$.

\begin{remark}
\label{rem-0}
We shall refer to the emergent bosonic fluctuations as to a mesoscopic description level, in between the microscopic one inherent to the quantum spin chain and the commutative one proper to macroscopic averages.
\end{remark}

The relations between the abstract Weyl algebra $\mathcal{W}(\chi,\sigma_\omega)$ and the local fluctuation operators is as follows. Generic local fluctuation operators are linear combinations of those of the $x_i\in\chi$: 
\begin{equation}
\label{qfa2}
(r,F_N):=\sum_{i=1}^dr_iF_N(x_i)\ ,
\end{equation}
For notational convenience, we shall also use
\begin{equation}
\label{qfa2b}
(r,F_N)=\frac{1}{\sqrt{N_T}}\sum_{k=-N}^{N}q_r^{(k)}=F_N(q_r)\ ,\quad 
q_r=\sum_{i=1}^dr_i\left(x_i-\omega(x_i)\right)\in M_p(\mathbb{C})\ ,
\end{equation}
and introduce local Weyl-like operators
\begin{equation}
\label{localWeyl}
W_N(r)={\rm e}^{i(r,F_N)}={\rm e}^{i\,F_N(q_r)}\qquad\forall r\in\mathbb{R}^d\ .
\end{equation}

For a system with normal quantum fluctuations, it is always possible to find a state on 
$\mathcal{W}(\chi,\sigma_\omega)$ with the properties of a normal Gaussian state \cite{Verbeure}. 

\begin{theorem}
\label{thm1}
If the system $\left(\omega,\chi\right)$ has normal quantum fluctuations, there exists a so-called quasi-free (also known as Gaussian) state $\Omega$ on $\mathcal{W}\left(\chi,\sigma_\omega\right)$  such that:
\begin{eqnarray}
\label{qfa1a}
\lim_{N\to+\infty}\omega\left(W_N(r)\right)&=&\Omega\left(W(r)\right)=
{\rm e}^{-\frac{1}{2}\,(r,\Sigma_\omega\,r)}\\ 
\label{qfa1b}
\lim_{N\to\infty}\omega\left(W_N(r_1)W_N(r_2)\right)&=&{\rm e}^{-(r_1+r_2,\Sigma_\omega r_1+r_2)/2\,-\,i\sigma_\omega(r_1,r_2)/2}
=\Omega\left(W(r_1)W(r_2)\right)\ ,
\end{eqnarray}
where $r_{1,2}\in\mathbb{R}^d$ and $\Sigma_\omega$ is a real symmetric $d\times d$ covariance matrix with entries
\begin{equation}
\label{qfa1aa}
\Sigma_\omega^{ij}=\lim_{N\to+\infty}\frac{1}{2}\omega\left(\left\{F_N(x_i)\,,\,F_N(x_j)\right\}\right)\ ,
\end{equation}
where $\left\{X\,,\,Y\right\}$ denotes the anti-commutator.
\end{theorem}

\begin{remark}
\label{rem0}
Actually, it can be proved \cite{Verbeure} that
$$
\lim_{N\to\infty}\left\|W_N(r_1)\,W_N(r_2)\,-\,W_N(r_1+r_2)\,{\rm e}^{-[(r_1,F_N)\,,\,(r_2,F_N)]/2}\right\|=0
$$
from which it follows that \eqref{qfa1b} can be extended to all products of Weyl-like operators, whence:
\begin{equation}
\label{rem0ref}
\lim_{N\to\infty}\omega\left(W_N(a)\,W_N(r)\,W_N(b)\right)=
\Omega\left(W(a)\,W(r)\,W(b)\right)
\end{equation}
for all $a,b,r\in\mathbb{R}^d$.
\end{remark}

Given the state $\Omega$ on the Weyl algebra $\mathcal{W}\left(\chi,\sigma_\omega\right)$, it turns out that one can construct a regular representation of the Weyl operators as acting on a suitable Hilbert space whereby they can be identified with exponentials of hermitian operators:
\begin{equation}
\label{qfa1c}
W(r) = e^{i\,(r,F)}\ ,\quad (r,F)=\sum_{i=1}^d\,r_iF(x_i)\ ,
\end{equation}
where $F$ is the $d$-dimensional vector of components $F(x_i)$ that
can thus be identified as the limits of the local quantum fluctuations operators 
$F_N(x_i)$.

The construction of the algebra of quantum fluctuation can be summarized by saying that Theorem \ref{thm1} justifies the following large $N$ correspondence 
$$
W_N(r)={\rm e}^{i(r,F_N)}\simeq W(r):={\rm e}^{i(r,F)}\ ,\quad 
\omega(W_N(r))\simeq\Omega(W(r))\ ,
$$ 
where $W_N(r)$ are local Weyl-like operators obtained as exponentials of local fluctuation operators and $W(r)$ are Weyl operators, while $\omega$ is the microscopic state on the spin chain quasi-local algebra $\mathcal{A}$ and $\Omega$ is a Gaussian state on the algebra generated by the $W(r)$ that satisfies equation \eqref{qfa1a}.

\section{Microscopic Dissipative Dynamics}

Given the algebra of quantum fluctuations, a relevant issue regards the mesoscopic dynamics inherited from a given microscopic time-evolution.
So far, only unitary microscopic dynamics have been considered and these have given rise to quasi-free, unitary
time-evolutions \cite{Verbeure}, namely to reversible dynamics transforming Weyl operators into Weyl operators.

Instead, in the following we shall focus upon a quantum spin chain undergoing an irreversible dissipative microscopic dynamics due to the presence of an environment to which the chain is weakly coupled.
This setting is typical of open quantum systems, thus the quantum spin chain is affected by decoherence due to noise and dissipation \cite{Alicki}.  

The main purpose of the following section is to show that, from a 
Lindblad-type microscopic dissipative dynamics which, in a way to be specified later on, is consistent with the chosen set of fluctuations, one obtains a mesoscopic quasi-free (also known as Gaussian) dissipative semigroup on the fluctuation level. 

Concretely, we shall study the fluctuation time-evolution emerging from a microscopic irreversible dynamics that, in the Heisenberg picture, corresponds to the local time-evolution equation
\begin{equation}
\label{dissdyn}
\partial_t X_t\,=\,\mathbb{L}_N[X_t]\ ,
\end{equation}
where $X\in\mathcal{A}_{[-N,N]}$ and $\mathbb{L}_N$ is a generator of Lindblad form, 
\begin{equation}
\mathbb{L}_N\left[X\right]=\mathbb{H}_N\left[X\right]+\mathbb{D}_N\left[X\right]\ .
\label{gen}
\end{equation}
Here, $\mathbb{H}_N\left[X\right]=i\left[H_N,X\right]$ is a translation invariant Hamiltonian term, where
\begin{equation}
\label{tih}
H_N=\sum_{k=-N}^{N}h^{(k)}\ ,
\end{equation}
with single site Hamiltonian $h=h^\dagger$, while $\mathbb{D}_N[X]$ is a dissipative term of the form:
\begin{eqnarray}
\label{lind1}
\mathbb{D}_N\left[X\right]&=&
\sum_{k,\ell=-N}^{N}J_{k\ell}\sum_{\mu,\nu=1}^{m}\,D_{\mu\nu}\,\Bigg(v_\mu^{(k)}\,X\,
(v_{\nu}^\dagger)^{(\ell)}\,-\,\frac{1}{2}\Big\{v_\mu^{(k)}\,(v_\nu^\dagger)^{(\ell)}\,,\,X\Big\}\Bigg)\\
\label{lind2}
&=&\sum_{k,\ell=-N}^{N}J_{k\ell}\sum_{\mu,\nu=1}^{m}\frac{D_{\mu\nu}}{2}\Bigg(\Big[v_\mu^{(k)}\,,\,X\Big]\,(v_{\nu}^\dagger)^{(\ell)}\,+\,v_\mu^{(k)}\,\Big[X\,,\,(v_\nu^\dagger)^{(\ell)}\Big]\Bigg)\ .
\end{eqnarray}
In the above expression, the operators $v_\nu$, $\nu=1,2,\ldots,m$, are single site Kraus operators and the coefficients $J_{k\ell}$ and $D_{\mu\nu}$ make for two matrices $J\geq 0$ and $D\geq 0$: the overall Kossakowski matrix $J\otimes D$ must be chosen  positive semi-definite in order to guarantee the complete positivity of the local dissipative time-evolution $\gamma_N^t=\exp(t\,\mathbb{L}_N)$  generated by $\mathbb{L}_N$ \cite{Alicki}.
The matrix $D$ encodes the information about the single site effects due to the presence of the environment, while the matrix $J$ accounts for the strength of these effects between different sites. 
Translational invariance of the dissipative contribution to the Lindblad generator is achieved by $J_{k\ell}$ such that:
\begin{equation}
\label{translinv}
J_{k\ell}\equiv J(k-\ell)\ ,\quad J_{kk}=J(0)>0\qquad\forall\,k,\ell\,\in\mathbb{Z}\ .
\end{equation}
We shall also assume that 
\begin{equation}
\label{assum0}
\sum_{\ell=-\infty}^{\infty}\left|J_{k\ell}\right|=\sum_{p=-\infty}^{\infty}\left|J(p)\right|<\infty\qquad\forall\,k\in\mathbb{Z}\ ,
\end{equation}
namely that the strength of the statistical coupling of far separated chain sites due to the environment be vanishingly small.

\subsection{Locality Conditions}
\label{loccondsec}

Our purpose in the following is to derive from the microscopic dissipative spin chain dynamics outlined above a mesoscopic dynamics for the large $N$ limits, $W(r)$, of the Weyl-like exponentials of local fluctuations, $W_N(r)={\rm e}^{i(r,F_N)}$, corresponding to the chosen set of single-site observables $x_i\in\chi$.
In general, the action of the local Lindblad generator $\mathbb{L}_N$ on $(r,F_N)=\sum_{j=1}^dr_jF_N(x_j)$ maps it out of the linear span of the elements of the set $\chi$, possibly generating operators that are linear combinations of tensor products supported by the whole interval $[-N,N]$. 

In order to recover, out of the action of $\mathbb{L}_N$, a mesoscopic dynamics for the Weyl algebra generated by the Weyl operators $W(r)$ we shall then assume that 
\begin{equation}
\label{assump0}
\mathbb{L}_N[x^{(k)}_i]=\sum_{j=1}^d\mathcal{L}_{ij}x_j^{(k)}\ ,\quad\mathcal{L}=\mathcal{H}+\mathcal{D}\ ,
\end{equation}
for all $x_i\in \chi$, $k\in\left[-N,N\right]$, where $\mathcal{L}$ is the $d\times d$ matrix with entries $\mathcal{L}_{ij}$ and $\mathcal{H}$, $\mathcal{D}$ are the $d\times d$ matrices with entries defined by
\begin{equation}
\label{assump1}
i\left[H_N,x_i^{(k)}\right]=
\sum_{j=1}^{d}\mathcal{H}_{ij}x_j^{(k)}\ ,\quad
\mathbb{D}_N\left[x_i^{(k)}\right]=\sum_{j=1}^{d}\mathcal{D}_{ij}x_j^{(k)} \ .
\end{equation} 
The above one is a request of both locality and consistency with the chosen set 
$\chi$; indeed, it amounts to asking that the linear span of $\chi$ is mapped into itself by $\mathbb{L}_N$. Moreover, in Appendix A, it is proved that \eqref{assump0} is equivalent to the following property of the generator.

\begin{lemma}
\label{lemma5}
Given a single-site matrix basis $\{o_\alpha\}_{\alpha=1}^{p^2}$ in $M_p(\mathbb{C})$ and a generator $\mathbb{L}_N$ satisfying 
$$
\mathbb{L}_N\left[o_\alpha^{(k)}\right]=\sum_{\beta=1}^{p^2}\,c_{\alpha\beta}^{k}\,o_\beta^{(k)}\ ,
$$
then:
$$
\mathbb{L}_N\left[o_{\alpha_1}^{(k_1)}o_{\alpha_2}^{(k_2)}\dots o_{\alpha_n}^{(k_n)}\right]=\sum_{\bar{\beta}}\,c^{\bar{k}}_{\bar{\alpha}\,\bar{\beta}}\, o_{\beta_1}^{(k_1)} o_{\beta_2}^{(k_2)}\dots o_{\beta_n}^{(k_n)} \ ,
$$
where the bar denotes multi-indices, {\it e. g.} $\bar{\alpha}=(\alpha_1,\alpha_2,\dots,\alpha_n)$, and  $c_{\bar{\alpha},\bar{\beta}}^{\bar{k}}$ are suitable coefficients.
\end{lemma}

\begin{remark}
\label{rem}
The local Lindblad generator $\mathbb{L}_N$ studied in \cite{BCF} has the form 
$$
\mathbb{L}_N[X_N]=i\Big[\sum_{k=0}^{N-1}h^{(k)}\,,\,X_N\Big]\,+\,\sum_{k,\ell=0}^{N-1}J_{k\ell}\sum_{\mu,\nu=1}^4\frac{D_{\mu\nu}}{2}\Big[\Big[v_\mu^{(k)}\,,\,X_N\Big],(v_{\nu}^\dag)^{(\ell)}\Big]\ .
$$
In this case, the reduction from the general Lindblad form \eqref{lind1} to the above is obtained by choosing $J_{k\ell}\equiv J(|k-\ell|)=J_{\ell k}$ and the matrix $D$ of the form 
$$
D=
\begin{pmatrix}
d&0&\gamma&\gamma\\
0&d&\gamma&\gamma\\
\gamma&\gamma&d&0\\
\gamma&\gamma&0&d
\end{pmatrix}\ ,\qquad d>0\ ,\quad |\gamma|\leq \frac{d}{2}\ .
$$
Then, it is the appearance of a double commutators that makes the generator obey \eqref{assump0}.
\end{remark}

The following result characterizes the action of local Lindbald generators satisfying \eqref{assump0}. 

\begin{proposition}
Given a set $\chi$ of single site operators and a generator $\mathbb{L}_N$ 
satisfying conditions \eqref{assum0} and \eqref{assump0}, the action of the generator on $W_N(r)={\rm e}^{i\left(r,F_N\right)}$ is such that
\begin{eqnarray}
\label{prop1a}
&&
\lim_{N\to+\infty}\Bigg\|\mathbb{L}_N\left[W_N(r)\right]-\Bigg(\frac{i}{\sqrt{N_T}}\sum_{k=-N}
^N\sum_{i,j=1}^{d}r_i\left(\mathcal{H}_{ij}+\mathcal{D}_{ij}\right)\,x_j\Bigg)\,
\,W_N(r)\\
&&\hskip 2cm
+\,\frac{1}{2}\Big[(r,F_N)\,,\,(r,(\mathcal{H}+\mathcal{D})F_N)\Big]\,\,W_N(r)\,-\,S(r;N)\,W_N(r)\Bigg\|=0\ ,
\label{thesis1}
\end{eqnarray}
where 
\begin{eqnarray}
\label{prop1b}
S(r;N)&=&\frac{1}{2}\Big(\mathbb{L}_N\left[(r,F_N)\right]\,(r,F_N)\,+\,(r,F_N)\,\mathbb{L}_N\left[(r,F_N)\right]\,-\,\mathbb{L}_N\left[(r,F_N)^2\right]\Big)\\ 
\nonumber
\mathbb{L}_N[(r,F_N)]&=&\frac{1}{\sqrt{N_T}}\sum_{k=-N}^N\sum_{i,j=1}^d r_i(\mathcal{H}_{ij}+\mathcal{D}_{ij})\,x_j^{(k)}\\
\label{prop1d}
&=&\Big(r,(\mathcal{H}+\mathcal{D})\,F_N\Big)\,+\,\sqrt{N_T}\,\Big(r,(\mathcal{H}+\mathcal{D})\,x_\omega\Big)\ ,
\end{eqnarray}
with $x_\omega\in\mathbb{R}^d$ a real vector with components 
$\omega(x_i)$, $x_i\in\chi$.
\label{action}
\end{proposition}

The proof is subdivided in two lemmas concerning the large $N$ approximation of the Hamiltonian and of the dissipative terms of the Lindblad generator: their proofs are given in Appendix B, respectively C.

\begin{lemma}
\label{lemma1}
For large $N$, the Hamiltonian action of the Lindblad generator can be approximated as follows:
\begin{equation}
i\left[H_N\,,\,W_N(r)\right]\simeq \left(i\frac{1}{\sqrt{N_T}}\sum_{k=-N}^{N}\sum_{i,j=1}^{d}r_i\mathcal{H}_{ij}x_j^{(k)}-\frac{1}{2}\Big[\left(r,F_N\right)\,,\,\left(r,\mathcal{H}F_N\right)\Big]\right)\,W_N(r)\ ,
\label{piece1}
\end{equation}
the error vanishing in norm.
\end{lemma}

\begin{lemma}
\label{lemma2}
For large $N$, the action of the dissipative part of the Lindblad generator can be approximated as follows:
\begin{equation}
\mathbb{D}_N\left[W_N(r)\right]\sim \left(i\frac{1}{\sqrt{N_T}}\sum_{k=-N}^{N}\sum_{i,j=1}^d\,r_i\mathcal{D}_{ij}x_j^{(k)}\,-\,\frac{1}{2}\Big[(r,F_N)\,,\,(r,\mathcal{D}F_N)\Big]\,+\,S(r;N)\right)W_N(r)\ ,
\end{equation}
where
$$
S(r;N)=\frac{1}{2}\Big(\mathbb{L}_N\left[(r,F_N)\right]\,(r,F_N)\,+\,(r,F_N)\,\mathbb{L}_N\left[(r,F_N)\right]\,-\,\mathbb{L}_N\left[(r,F_N)^2\right]\Big)\ ,
$$
the error vanishing in norm.
\end{lemma}

\subsection{Time-invariant microscopic state}

In the next section we study the time-evolution of the mesoscopic degrees of freedom corresponding to the microscopic equation \eqref{dissdyn}, assuming  the microscopic state $\omega$ to be left invariant by the dissipative dynamics generated by $\mathbb{L}_N$ and formally represented by the semigroup of local maps 
$\Phi^N_t=\exp(t\,\mathbb{L}_N)$, $t\geq 0$. Namely, we shall assume
\begin{equation}
\label{invstate}
\omega\Big(\Phi^N_t[X]\Big)=\omega(X)\,\Leftrightarrow\,\omega\Big(\mathbb{L}_N[X]\Big)=0\ ,
\end{equation} 
for all $X$ in local algebras $\mathcal{A}_{[-N,N]}$.

\begin{remark}
\label{rem2}
The local dynamical maps $\Phi^N_t$ satisfy the forward in time composition law, typical of irreversible time-evolutions:
\begin{equation}
\label{semigroup}
\Phi^N_t\circ\Phi^N_s=\Phi^N_s\circ\Phi^N_t=\Phi^N_{t+s}\qquad\forall s,t\geq 0\ .
\end{equation}
The request of microscopic time-invariance of $\omega$ is essential to get a mesoscopic fluctuation dynamics with the semigroup property.
For $\omega$ not satisfying \eqref{invstate}, the dissipative fluctuation dynamics 
would result explicitly time-dependent and non-Markovian \cite{BCFN}.
\end{remark}

The consequences of a time-invariant $\omega$ can be appreciated by  considering the expectation of the action of the Lindblad generator on a fluctuation operator. Recalling
\eqref{prop1d} in Proposition \ref{action}, one gets
$$
\mathbb{L}_N[(r,F_N)]=\Big(r,(\mathcal{H}+\mathcal{D})F_N\Big)\,+\,\sqrt{N_T}\,\Big(r,(\mathcal{H}+\mathcal{D})\,x_\omega\Big)\ ,
$$
where the last quantity is a scalar multiple of the identity operator.
Therefore, since fluctuation operators have vanishing mean values, $\omega(F_N(x_i))=0$,
time-invariance of $\omega$ yields
\begin{eqnarray}
\label{timeinv1}
\omega\Big(\mathbb{L}_N[(r,F_N)]\Big)&=&0\,\Longrightarrow\,\sqrt{N_T}\,\Big(r,(\mathcal{H}+\mathcal{D})\,x_\omega\Big)=0\qquad\hbox{so that}\\
\label{timeinv2}
\mathbb{L}_N[(r,F_N)]&=&\Big(r,(\mathcal{H}+\mathcal{D})F_N\Big)\quad\hbox{and}\quad
\Phi^N_t[(r,F_N)]=\Big(r,{\rm e}^{t(\mathcal{H}+\mathcal{D})}\,F_N\Big)\ . 
\end{eqnarray}

Furthermore, consider the quantity $S(r;N)$ in \eqref{prop1b}; from \eqref{invstate} it follows that
$$
\omega(S(r;N))=\frac{1}{2}\omega\Big(\mathbb{L}[(r,F_N)]\,(r,F_N)\Big)\,+\,
\frac{1}{2}\omega\Big((r,F_N)\,\mathbb{L}[(r,F_N)]\Big)\ .
$$
Using the fact that local quantum fluctuations have vanishing mean values with respect to $\omega$ (see \eqref{lqf0}) and by means of \eqref{prop1d}, one gets
\begin{eqnarray*}
\omega\Big(\mathbb{L}[(r,F_N)]\,(r,F_N)\Big)&=&\sum_{i,j,k=1}^dr_ir_j\,
\Big(\mathcal{H}_{ik}+\mathcal{D}_{ik}\Big)\,\omega\Big(F_N(x_k)F_N(x_j)\Big)\\
\omega\Big((r,F_N)\,\mathbb{L}[(r,F_N)]\Big)&=&\sum_{i,j,k=1}^dr_ir_j\,
\Big(\mathcal{H}_{ik}+\mathcal{D}_{ik}\Big)\,\omega\Big(F_N(x_j)F_N(x_k)\Big)\\
\end{eqnarray*}
Then, in the limit $N\to+\infty$, the fluctuation relation \eqref{qfa1aa} yields
\begin{equation}
\label{MFF1}
\lim_{N\to+\infty}\omega(S(r;N))=\Big(r,(\mathcal{H}+\mathcal{D})\,\Sigma_\omega\,r\Big)\ ,
\end{equation}
where $\Sigma_\omega$ is the fluctuation covariance matrix.

The proof of Proposition \ref{action} shows that $S(r;N)$ is the only operator involving tensor products of operators from more than one site produced by the action of a Lindblad generator with the property \eqref{assump0} on local fluctuations. 
It is the sum of terms of the form 
\begin{equation}
\label{nonlocal}
R_N=\frac{1}{N_T}\sum_{k,\ell=-N}^N\,J_{k\ell}\,\,a^{(k)}\,b^{(\ell)}\ ,
\end{equation}
with $a$ and $b$ suitable single-site operators. Notice that $\lim_{N\to+\infty}\omega(R_N)$ always exists because of the assumption \eqref{assum0}.

Despite the fact that $S(r;N)$ does not in general scale as a mean-field observable, we already know that the time-invariance of $\omega$ implies the convergence of its average in the large $N$ limit.
The following lemma, whose proof is reported in Appendix D, exhibits further properties of this kind of observables.

\begin{lemma}
\label{lemma4}
Given a set of coefficients $J_{k\ell}$ such that
$$
J_{k\ell}\equiv J(k-\ell)=J^*_{\ell k}\ ,\qquad \sum_{\ell=-\infty}^{\infty}\left|J_{k\ell}\right|=\sum_{p=-\infty}^{\infty}\left|J(p)\right|<\infty\ ,
$$ 
and a translation invariant clustering state $\omega$, then 
$$
\lim_{N\to+\infty}\left\|\left[W_N(r)\,,\,R_N\right]\right\|=0\quad \hbox{and}\quad
\lim_{N\to+\infty}\omega\left(\left(R_N-R\right)^\dagger \left(R_N-R\right)\right)=0
$$
for all local Weyl-like operators as in \eqref{localWeyl} and $R_N$ as in \eqref{nonlocal} with
$\displaystyle R:=\lim_{N\to+\infty}\omega(R_N)$.
\label{mean2}
\end{lemma}

\section{Mesoscopic Dissipative Dynamics}

In this section we shall show that, under the locality condition \eqref{assump0} on the microscopic Lindblad generator, 
the mesoscopic dynamics that emerges in the limit $N\to+\infty$ is described by  a
semi-group $\{\Phi_t\}_{t\geq0}$ of completely positive, unital maps 
on the quantum fluctuation algebra.

We shall prove that the microscopic dissipative dynamical maps $\Phi_t^N$ on the local algebras $\mathcal{A}_{[-N,N]}$ define mesoscopic dynamical maps $\Phi_t$ on the Weyl algebra $\mathcal{W}(\chi,\sigma_\omega)$ of quantum fluctuations through the limits
\begin{equation}
\label{RD2}
\lim_{N\to+\infty}\omega\left(W_N(a)\Phi^N_t\left[W_N(r)\right]\,
W_N(b)\right)=\Omega\Big(W(a)\,\Phi_t\left[W(r)\right]\,W(b)\Big)\ ,
\end{equation}
for all Weyl-like operators $W_N(a)$, $W_N(b)$, $W_N(r)$, with $W(a)$, 
$W(b)$ and $W(r)$ the corresponding limiting Weyl operators, with  $\omega$ and $\Omega$ the microscopic state on the quantum chain, respectively the mesoscopic state on the Weyl algebra defined by \eqref{qfa1a}. 

\begin{remark}
\label{remfinal}
The above type of convergence reduces to \eqref{rem0ref} in Remark \ref{rem0} for $t=0$. Moreover, it conforms to the fact that the action of any map on the quantum fluctuation algebra is specified by its action on the Weyl operators 
$W(r)$. As the dynamical maps $\Phi_t$ we are looking for have to transform the Weyl algebra into itself, their action is in turn completely defined by correlation functions of the type in \eqref{RD2}.
Furthermore, both the state $\Omega$ and the Weyl operators $W(r)$ arise from the large $N$ limit of the microscopic local Weyl-like operators $W_N(r)$ with respect to the microscopic state $\omega$.
\end{remark}

We will look for dynamical maps $\Phi_t$ that be quasi-free, namely  that map Weyl operators into Weyl operators:
\begin{equation}
\label{qfm0}
\Phi_t[W(r)]={\rm e}^{f_r(t)}\,W(r_t)\qquad\forall\, r\in\mathbb{R}^d\ ,
\end{equation}
where both the time-dependent function $f_r(t)$ and vector $r_t\in\mathbb{R}^d$  
are unknowns to be determined.
The maps are unital, $\Phi_t[1]=1$, and must be completely positive. As such they must obey the Schwartz positivity inequality \cite{Alicki}
\begin{equation}
\label{Schwpos}
\Phi_t\left[X^\dag X\right]\,\geq\,\Phi_t\left[X^\dag\right]\,\Phi_t\left[X\right]\qquad\forall X\in \mathcal{W}(\chi,\sigma_\omega)\ .
\end{equation}
Then, since the Weyl operators $W(r)$ are unitary, $f_r(t)$ must satisfy
\begin{equation}
\label{Schwpos1}
\|\Phi_t\left[W(r)\right]\|=\left|{\rm e}^{f_r(t)}\right|\leq \|\Phi_t[1]\|=1\ .
\end{equation}

The proof of equation \eqref{RD2} will be based on a family of local microscopic maps 
$\Psi_t$ on the microscopic quantum chain that interpolate between the microscopic, $\Phi^N_t$, and the mesoscopic dissipative time-evolution, $\Phi_t$. They are defined by:
\begin{eqnarray}
\label{intmap0}
W_N(r)\mapsto\Psi_t\left[W_N(r)\right]&=&{\rm e}^{f_r(t)}\,W_N(r_t)\,=\,
{\rm e}^{f_r(t)}\,{\rm e}^{i(r_t,F_N)}\\
\label{intmap1}
r_t&=&\mathcal{X}^{tr}_t\,r\ ,\quad \mathcal{X}_t={\rm e}^{t\,\left(\mathcal{D}+\mathcal{H}\right)}\\
\label{intmap2}
f_r(t)&=&-\left(r,\mathcal{Y}_t r\right)\ ,\quad \mathcal{Y}_t=\frac{1}{2}\left(\Sigma_\omega\,-\,\mathcal{X}_t\,\Sigma_\omega\, \mathcal{X}_t^{tr}\right)\ ,
\end{eqnarray}
where $\mathcal{X}^{tr}$ denotes the transposition of the $d\times d$ matrix $\mathcal{X}$ and $\Sigma_\omega$ is the fluctuation covariance matrix.
Because of the convergence property discussed in Remark \ref{rem0}, we know that
$$
\lim_{N\to+\infty}\omega\left(W_N(a)\,W_N(r_t)\,W_N(b)\right)=\Omega\Big(W(a)\,W(r_t)\,W(b)\Big)\ .
$$
We are going to show that the mesoscopic dynamical maps $\Phi_t$ in \eqref{RD2}
are such that
\begin{equation}
\label{final}
\Phi_t\left[W(r)\right]={\rm e}^{f_r(t)}\,W(r_t)\qquad\forall\ W(r)\in\mathcal{W}(\chi,\sigma_\omega)\ , 
\end{equation}
where $f_r(t)$ and $r_t$ are given by \eqref{intmap1} and \eqref{intmap2}.

The maps $\Phi_t$ compose as a semigroup; indeed, for all $s,t\geq 0$,
\begin{eqnarray*}
\Phi_s\circ\Phi_t\left[W(r)\right]&=&{\rm e}^{-(r,\mathcal{Y}_t\,r)-(r_t,\mathcal{Y}_s\,r_t)}\,W((r_t)_s)\\
&=&{\rm e}^{-(r,\mathcal{Y}_t\,r)-(r,\mathcal{X}_t\mathcal{Y}_s\mathcal{X}_t^{tr}\,r)}\,W(r_{t+s})\\
&=&{\rm e}^{-(r,\mathcal{Y}_{t+s}\,r)}\,W(r_{t+s})=\Phi_{t+s}\left[W(r)\right]\ .
\end{eqnarray*}
Furthermore, as required by complete positivity and unitality, and proved by the following lemma, the function
$f_r(t)$ defined by \eqref{intmap2} is such that $\exp(f_r(t))\leq 1$.

\begin{lemma}
\label{lemma6}
The invariance of the microscopic state $\omega$  with respect to the microscopic dissipative dynamics $\Phi^N_t$ implies that $\mathcal{Y}_t$ is negative semi-definite so that $f_r(t)\leq 0$.
\end{lemma}

\begin{proof}
We shall show that $\omega\circ\Phi^N_t=\omega$, $t\geq 0$, makes negative semi-definite, $\mathcal{Y}_t\leq 0$, the matrix defined by \eqref{intmap2}, for all $t\geq 0$.
Let $\lambda\in\mathbb{C}^d$ be a generic complex vector and set $q_\lambda=\sum_{j=1}^d
\lambda_j\,x_j$; then, using Schwartz positivity, time-invariance and the second relation in \eqref{timeinv2}, one estimates 
\begin{eqnarray*}
&&
\frac{1}{2}\sum_{i,j=1}^d\lambda_i^*\lambda_j\,\omega\Big(\Big\{F_N(x_i)\,,\,F_N(x_j)
\Big\}\Big)=\frac{1}{2}\sum_{i,j=1}^d\lambda_i^*\lambda_j\,\omega\Big(\Phi^N_t\Big[\Big\{F_N(x_i)\,,\,F_N(x_j)\Big\}\Big]\Big)\\
&&\hskip 2cm
\geq\frac{1}{2}\omega\Big(\Phi^N_t\left[F_N(q^\dag_\lambda)\right]\,
\Phi_t^N\left[F_N(q_\lambda)\right]\Big)\,+\,\frac{1}{2}\omega\Big(\Phi^N_t\left[F_N(q_\lambda)\right]\,\Phi^N_t[F_N(q^\dag_\lambda)]\Big)\\ 
&&\hskip 2cm
=\frac{1}{2}\omega\Big(\Big(\lambda,\mathcal{X}_t\,F_N\Big)\,\Big(\lambda^*,\mathcal{X}_t\,F_N\Big)\Big)\,+\,\frac{1}{2}\omega\Big(\Big(\lambda^*,\mathcal{X}_t\,F_N\Big)\,\Big(\lambda,\mathcal{X}_t\,F_N\Big)\Big)\\
&&\hskip2 cm
=\frac{1}{2}\sum_{i,j;r,s=1}^d\lambda_i^*\lambda_r\mathcal{X}_t^{ij}\,\mathcal{X}_t^{rs}\omega\Big(\Big\{F_N(x_j)\,,\,F_N(x_s)\Big\}\Big)\ .
\end{eqnarray*}
In the large $N$ limit one thus obtain, for all $\lambda\in\mathbb{C}^d$,
$$
\Big(\lambda,\Sigma_\omega\,\lambda\Big)\geq \sum_{i,j;r,s=1}^d\lambda_i^*\lambda_r\mathcal{X}_t^{ij}\,\mathcal{X}_t^{rs}\,\Sigma^{js}_\omega=
\Big(\lambda,\mathcal{X}_t\Sigma_\omega\mathcal{X}^{tr}_t\,\lambda\Big)\ .
$$
\end{proof}

\begin{remark}
\label{rem3}
As the dynamical maps $\Phi_t$ tramsform Weyl operators into Weyl operators, their dual maps acting on generic
states $\Omega$ on the Weyl algebra sending them into $\Omega_t=\Omega\circ\Phi_t$, transform  Gaussian states into Gaussian states. 
For instance, as it emerges from a microscopic time invariant state $\omega$, the state $\Omega$ in \eqref{qfa1a} and \eqref{qfa1b} is left invariant by $\Phi_t$:
\begin{eqnarray*}
\Omega\left(\Phi_t\left[W(r)\right]\right)&=&{\rm e}^{f_r(t)}\,\Omega\left(W(r_t)\right)=
{\rm e}^{-\frac{1}{2}(r,\mathcal{Y}_tr)-\frac{1}{2}\,(r_t,\Sigma_\omega\,r_t)}\\ 
&=&
{\rm e}^{-\frac{1}{2}(r,\mathcal{Y}_t r)-\frac{1}{2}\,(r,\mathcal{X}_t\Sigma_\omega\,\mathcal{X}_t^{tr}r)}={\rm e}^{-\frac{1}{2}\,(r,\Sigma_\omega\,r)}=\Omega\left(W(r)\right)\ .
\end{eqnarray*}
As they inherit the semigroup property from the $\Phi_t$, the dual maps have a generator and this generator must then be at most quadratic in the mesoscopic operators $F(x_i)$ arising from the local fluctuation operators $F_N(x_i)$. The explicit form of the generator is derived by duality and by explicitly computing the time-derivative of $\Phi_t\left[W(r)\right]$, using the Weyl algebraic relations to reconstruct it by means of the action $\mathbb{L}\left[W(r)\right]$: it turns out that the resulting Kossakowski matrix 
is positive semi-definite so that the maps $\Phi_t$ on $\mathcal{W}(\chi,\sigma_\omega)$ are completely positive.
\end{remark}

In conclusion, in order to prove \eqref{final}, we need show that
$$
\lim_{N\to+\infty}\omega\left(W_N(a)\Big(\Psi_t\left[W_N(r)\right]\,-\,\Phi^N_t\left[W_N(r)\right]\Big)W_N(b)\right)=0\ .
$$
Actually, like all positive, normalised linear functionals on the Weyl algebra, $\omega$ satisfies the Cauchy-Schwartz inequality 
$|\omega(a^\dag\,b)|^2\leq\omega(a^\dag a)\,\omega(b^\dag b)$, whence the unitarity of the Weyl-like operators $W_N(r)$ yields 
\begin{eqnarray}
\label{CSIN1}
\Big|\omega\left(W_N(a)\,\Delta_N(t,r)\,W_N(b)\right)\Big|^2&\leq& \omega\Big(W_N(a)\Delta_N(t,r)\Delta^\dag_N(t,r)\,W^\dag_N(a)\Big)\\
\label{CSIN2}
\Delta_N(t,r)&=&\Psi_t[W_N(r)]\,-\,\Phi^N_t\left[W_N(r)\right]\ .
\end{eqnarray} 

In order to show that the right hand side of the above inequality vanishes with $N\to+\infty$, we relate the interpolating map $\Psi_t$ to the local microscopic dissipative dynamics $\Phi^N_t={\rm e}^{t\,\mathbb{L}_N}$; namely, we study 
the time-derivative of $\Psi_t$ and its relations with the generator $\mathbb{L}_N$. 
The structure of the time derivative can be derived by means of the following lemma whose proof can be found in Appendix E.

\begin{lemma}
\label{applemma}
Let $M_t$ be a time-dependent Hermitean matrix and $\displaystyle N_t={\rm e}^{iM_t}$. 
Then,
\begin{equation}
\label{appb1}
\dot{N_t}:=\frac{{\rm d}N_t}{{\rm d}t}=O_t\,N_t\ ,\quad O_t:=\sum_{k=1}^\infty\frac{i^k}{k!}\mathbb{K}_{M_t}^{k-1}[\dot{M}_t]\ ,
\end{equation}
where 
$\mathbb{K}_{M_t}[\dot{M}_t]=\Big[M_t\,,\,\mathbb{K}_{M_t}^{n-1}[\dot{M}_t]\Big]$ and $\mathbb{K}^0_{M_t}[\dot{M}_t]=\dot{M}_t$.
\end{lemma}

Equipped with this result, we can show that, for large $N$, all terms in the series expansion of $\displaystyle\frac{{\rm d}}{{\rm d}t}\Psi_t[W_N(r)]$ of order larger than $2$ vanish in norm.

\begin{proposition}
\label{prop2}
For large $N$, the behaviour of $\displaystyle\frac{{\rm d}}{{\rm d}t}\Psi_t[W_N(r)]$
can be approximated by
\begin{eqnarray}
\nonumber
\frac{{\rm d}}{{\rm d}t}\Psi_t[W_N(r)]&\simeq&\Bigg(i\left(r_t,(\mathcal{H}\,+\,\mathcal{D})F_N\right)\,-\,\frac{1}{2}\left[(r_t,F_N)\,,\,(r_t,(\mathcal{H}\,+\,\mathcal{D})F_N)\right]\,+\\
&+&\left(r_t,\left(\mathcal{H}\,+\,\mathcal{D}\right)\Sigma_\omega\,r_t\right)\Bigg)\,\Psi_t\left[W_N(r)\right]\ ,
\end{eqnarray}
the error vanishing in norm.
\end{proposition}

\begin{proof}
The time-derivative
$$
\frac{{\rm d}}{{\rm d}t}\Psi_t\left[W_N(r)\right]=\frac{{\rm d}f_r(t)}{{\rm d}t}\,\Psi_t\left[W_N(r)\right]\,+\,{\rm e}^{f_r(t)}\,\frac{{\rm d}}{{\rm d}t}W_N(r_t)$$
consist of two terms: from \eqref{intmap2}, by direct computation, the first one contains
\begin{equation}
\label{timederf}
\dot{f}_r(t)=\frac{1}{2}\,\Big(r_t\,,\,\Big((\mathcal{H}+\mathcal{D})\Sigma_\omega
\,+\,\Sigma_\omega(\mathcal{H}+\mathcal{D})^{tr}\Big)\,r_t\Big)=
\Big(r_t\,,\,(\mathcal{H}+\mathcal{D})\Sigma_\omega\,r_t\Big)\ ,
\end{equation}
where in the last equality use has been made of the reality of the vector $r_t=(r_t^1,\ldots,r^d_t)^{tr}$ and of the fact that the covariance matrix is real symmetric, namely $\Sigma_\omega=\Sigma_\omega^{tr}$.
Using Lemma \ref{applemma} and \eqref{qfa2b}, the second term contains
\begin{eqnarray*}
\frac{{\rm d}}{{\rm d}t}W_N(r_t)&=&\frac{{\rm d}}{{\rm d}t}{\rm e}^{i\,F_N(q_{r_t})}\,=\,\Big(i\,F_N\left(\dot{q}_{r_t}\right)-\frac{1}{2}\Big[F_N\left(q_{r_t}\right)\,,\,F_N\left(\dot{q}_{r_t}\right)\Big]\Big]\Big)\,W_N(r_t)\\ 
&+&\sum_{n=3}^{+\infty}\frac{i^n}{n!}\,\mathbb{K}^{n-1}_{F_N(q_{r_t})}\left[F_N(\dot{q}_{r_t})\right]\,W_N(r_t)\ ,
\end{eqnarray*}
where $\dot{q}_{r_t}=(\dot{r}_t,F_N)=\Big(r_t,(\mathcal{H}+\mathcal{D})F_N\Big)$. Then, since operators at different sites commute, one estimates
\begin{eqnarray*}
\left\|\sum_{n=3}^{+\infty}\frac{i^n}{n!}\,\mathbb{K}^{n-1}_{F_N(q_{r_t})}\left[F_N(\dot{q}_{r_t})\right]\,W_N(r_t)\right\|&\leq&
\left\|\sum_{n=3}^{+\infty}\frac{1}{n!}\,\sum_{k=-N}^N\frac{1}{N_T^{n/2}}\,
\mathbb{K}^{n-1}_{q_{r_t}}\left[\dot{q}^{(k)}_{r_t}\right]\right\|\\
&\leq&\frac{1}{\sqrt{N_T}}\,\sum_{n=3}^{+\infty}\frac{(2\|q_{r_t}\|)^{n-1}}{n!}\,\|\dot{q}_{r_t}\|\leq \frac{{\rm e}^{2\|q_{r_t}\|}}{\sqrt{N_T}}\,\|\dot{q}_{r_t}\|\ . 
\end{eqnarray*}
The result thus follows as $q_{r_t}$ and $\dot{q}_{r_t}$ are bounded single-site operators for all $t\geq 0$ belonging to finite intervals of time.
\end{proof}

According to the previous discussion, the convergence of the microscopic dissipative dynamics $\Phi^N_t={\rm e}^{t\,\mathbb{L}_N}$ to the mesoscopic dissipative dynamics $\Phi_t$ in \eqref{qfm0} amounts to the validity of the following result.

\begin{theorem}
\label{theorem}
Given a quantum spin chain with normal quantum fluctuations $(\omega,\chi)$, $\chi$ a finite set of single-site observables and a local Lindblad generator staisfying assumption \eqref{assump0} and preserving the microscopic state $\omega$, then
$$
\lim_{N\to+\infty}\omega\left(W_N(a)\,\Delta_N(t,r)\,\Delta^\dag_N(t,r)\,W^\dag_N(a)\right)=0\ ,
$$
where $\displaystyle\Delta_N(t,r)=\Psi_t\left[W_N(r)\right]\,-\,\Phi^N_t\left[W_N(r)\right]$ and $\Psi_t$ is defined as in \eqref{intmap0}--\eqref{intmap2}. 
\end{theorem}

\begin{proof}
Notice that
\begin{eqnarray*}
\Psi_t\left[W_N(r)\right]-\Phi_t^N\left[W_N(r)\right]&=&\int_0^tdy\frac{d}{dy} e^{(t-y)\mathbb{L}_N}\left[\Psi_y\left[W_N(r)\right]\right]\\
&=&\int_0^tdy\ {\rm e}^{(t-y)\mathbb{L}_N}\left[\frac{d}{dy}\Psi_y\left[W_N(r)\right]\,-\,\mathbb{L}_N\left[\Psi_y\left[W_N(r)\right]\right]\right]\ .
\end{eqnarray*}
From Propositon \ref{action}, using \eqref{timeinv1} and \eqref{timeinv2}, for large $N$, one approximates 
$$
\mathbb{L}_N\left[W_N(r_y)\right]\simeq \Bigg(i\left(r_y,\left(\mathcal{H}+\mathcal{D}\right)F_N\right)\,-\,\frac{1}{2}\left[(r_y,F_N)\,,\,\Big(r_y,(\mathcal{H}+\mathcal{D})F_N\Big)\right]\,+\,S(r_y;N)\Bigg)\,W_N(r_y)\ .
$$
On the other hand, Proposition \ref{applemma} asserts that the time derivative can be approximated as follows
\begin{eqnarray*}
\frac{d}{dy}\Psi_y\left[W_N(r)\right]&\simeq&\Bigg(i\left(r_y,(\mathcal{H}+\mathcal{D})F_N\right)-\frac{1}{2}\left[(r_y,F_N)\,,\,(r_y,(\mathcal{H}+\mathcal{D})F_N)\right]+\\
&+&\Big(r_y,\left(\mathcal{H}+\mathcal{D}\right)\Sigma_\omega\,r_y\Big)\Bigg)\,\Psi_y\left[W_N(r)\right]\ .
\end{eqnarray*}
Since the errors in these approximations vanish in norm for all finite $t\geq 0$ and $\Phi_t^N$ is a contracting map, $\|\Phi^N_t[a^\dag a]\|\leq\|a\|^2$, what
remains to be controlled is the quantity
\begin{eqnarray*}
&&\omega\left(W_N(a)\Delta_N(t,r)\,\Delta^\dag_N(t,r)\,W^\dag_N(a)\right)=\\
&&\hskip 1cm =\int_0^t{\rm d}y\int_0^t{\rm d}z\ \omega\left(W_N(a)\Phi^N_{t-y}\left[D(r_y;N)\right]\,
\Phi^N_{t-z}\left[D^\dag(r_z;N)\right]\,W^\dag_N(a)\right)\\
&&\hskip 1cm
D(r_y;N):=Z(r_y;N)\Psi_y^N\left[W_N(r)\right]\\ 
&&
\hskip 1cm
Z(r_y;N):=\Big(r_y,\left(\mathcal{H}+\mathcal{D}\right)\Sigma_\omega\,r_y\Big)-S(r_y;N)\ .
\end{eqnarray*}
Using the Cauchy-Schwarz inequality \eqref{CSIN1} and then twice the Schwartz positivity inequality \eqref{Schwpos}, once for $\Phi^N_t$ and the other one for $\Psi_t$, the proof of the theorem reduces to controlling
\begin{eqnarray*}
\omega\left(W_N(a)\Phi^N_{t-y}\left[D(r_y;N)\,D^\dag(r_y;N)\right]\,W_N(b)\right)&\leq&
{\rm e}^{2f_r(y)}\,\omega\left(W^\dag_N(b)\,\Phi^N_{t-y}\left[Z^2(r_y;N)\right]\,W_N(b)\right)\\
&\leq&
\omega\left(W^\dag_N(b)\,\Phi^N_{t-y}\left[Z^2(r_y;N)\right]\,W_N(b)\right)\ .
\end{eqnarray*}
Indeed, $Z_N(r_y;N)$ is hermitian and ${\rm e}^{2f_r(y)}\leq 1$ (see Lemma \ref{lemma6}).
Consider now
$$
Z^2(r_y;N)=\Big(S(r_y;N)\,-\,\dot{f}_r(y)\Big)^2\ ,
$$
where use has been made of \eqref{timederf}.
The operator $S(r_y;N)$ is of the form \eqref{nonlocal}; then, Lemma \ref{lemma5} implies that its support is not altered by the local dissipative dynamics $\Phi^N_t$ so that 
$$
\Phi^N_{t-y}\left[S(r_y;N)\right]=\frac{1}{N_T}\sum_{k,\ell=-N}^{N}\,J_{k\ell}\,\sum_{\alpha,\beta=1}^{p^2}\,c^{\alpha\beta}_{k\ell}(t-y)\,o_\alpha^{(k)} o_\beta^{(\ell)}\ ,
$$
where $c^{\alpha\beta}_{k\ell}(t)$ are suitable coefficients, bounded for all finite $t\geq 0$, and $\{o_\alpha\}_{\alpha=1}^{p^2}$ is a single-site matrix basis in the algebra $M_p(\mathbb{C})$.
Analogously, 
$$
\Phi^N_{t-y}\left[S^2(r_y;N)\right]=\frac{1}{N^2_T}\sum_{k,\ell;p,q=-N}^{N}\,\sum_{\alpha,\beta;\mu,\nu=1}^{p^2}\,J_{k\ell}\,J_{pq}\,d^{\alpha\beta,\mu\nu}_{k\ell,pq}(t-y)\, o_\alpha^{(k)} o_\beta^{(\ell)}\, o_{\mu}^{(p)} o_{\nu}^{(q)}\ ,
$$
with $d^{\alpha\beta,\mu\nu}_{k\ell,pq}(t)$ bounded coefficients for all finite positive times $t\geq 0$.

Lemma \ref{mean2} asserts
that the sums over $k,\ell$ and $p,q$ commute with the Weyl-like operators $W_N(r)$ when $N\to+\infty$; then, using the time-invariance under $\Phi^N_t$ of the microscopic state $\omega$, we get:
\begin{eqnarray*}
&&\hskip-1cm
\lim_{N\to+\infty}\omega\left(W_N(a)\Phi^N_t\left[Z^2(r_y;N)\right]\,W^\dag_N(a)\right)=\lim_{N\to+\infty}\omega\left(\Phi^N_t\left[Z^2(r_y;N)\right]\,W_N(a)W_N^\dag(a)\right)\\
&&\hskip 1cm
=\lim_{N\to+\infty}\omega\left(\Phi^N_t\left[Z^2(r_y;N)\right]\right)=
\lim_{N\to+\infty}\omega\left(Z^2(r_y;N)\right)\ .
\end{eqnarray*}
The proof is thus completed by means of Lemma \ref{mean2}, of \eqref{MFF1} and of \eqref{timederf} which imply
$$
\lim_{N\to+\infty}\omega\left(S^2(r;N)\right)=\Big(\lim_{N\to+\infty}\omega\left(S(r;N)\right)\Big)^2=\dot{f}^2_r(t)\ .
$$
\end{proof}

\subsection{Application}

We illustrate the previous results by means of the following model consisting of 
a doubly infinite chain of spin $1$ systems. Each lattice site supports the algebra of complex $3\times 3$ matrices generated by the angular momentum operators
$J_{1,2,3}$ such that 
$$
[J_\mu,J_\nu]=i\,\varepsilon_{\mu\nu\delta}\,J_\delta\ ,\quad \sum_{\mu=1}^3J_\mu^2=2\ ,
$$
with $\varepsilon_{\mu\nu\delta}$ the totally anti-symmetric $3$-tensor.
In order to explicitly compute the quantities which follow, it is convenient to make use of
the orthonormal basis of $\mathbb{C}^3$ consisting of the eigenvectors of $J_3$:
$J_3\vert \mu\rangle=\mu\vert\mu\rangle$,
$\mu=-1,0,1$. Then, by using the raising and lowering operators $J_\pm=J_1\pm i\,J_2$ such that $J_+\vert1\rangle=J_-\vert-1\rangle=0$, $J_\pm\vert0\rangle=\sqrt{2}
\vert\pm1\rangle$, one computes
\begin{equation}
\label{spin1+-}
J_1\vert \pm1\rangle=\frac{1}{\sqrt{2}}\vert 0\rangle\ ,\ J_1\vert0\rangle=\frac{\vert 1\rangle+\vert-1\rangle}{\sqrt{2}}\ ;\ 
J_2\vert \pm1\rangle=\pm\frac{i}{\sqrt{2}}\vert 0\rangle\ ,\ J_2\vert0\rangle=\frac{\vert 1\rangle-\vert-1\rangle}{i\sqrt{2}}\ .
\end{equation}

We equip the quasi-local algebra generated by the local algebras supported by finitely many sites with
the translation invariant, factorized state $\omega_\beta$ such that, as one readily computes using the chosen orthonormal basis,
\begin{equation}
\label{Gibbsspin1}
\omega_1=\omega_2=0\ ,\ 
\omega_3=-2\frac{\sinh(\beta\omega)}{1+2\cosh(\beta\omega)}\quad \hbox{where}\quad
\omega_\mu:=\omega_\beta(J^{(k)}_\mu)=\frac{\Tr\left({\rm e}^{-\beta\omega J_3}J_\mu\right)}{1+2\cosh(\beta\omega)}
\  .
\end{equation}

It corresponds to a doubly infinite tensor product of a same Gibbs state at inverse temperature $\beta$ with Hamiltonian $\omega J_3$.
This state is evidently clustering so that mean-field observables behave as scalar multiples of the identity when $N$ becomes large.

We shall consider a set $\chi$ consisting of $J_{1,2,3}$ and their local fluctuations
\begin{equation}
\label{flucspin1}
F_N(J_{1,2})=\frac{1}{\sqrt{2N+1}}\sum_{k=-N}^NJ^{(k)}_{1,2}\ ,\quad
F_N(J_3)=\frac{1}{\sqrt{2N+1}}\sum_{k=-N}^N\left(J^{(k)}_3-\omega_3\right)\ .
\end{equation}
Since spin operators at different sites commute, one has
$$
\left[F_N(J_\mu)\,,\,F_N(J_\nu)\right]=i\,\varepsilon_{\mu\nu\delta}\,\frac{1}{2N+1}\sum_{k=-N}^NJ^{(k)}_\delta\ .
$$
Since the state $\omega_\beta$ is translation-invariant, 
$\displaystyle
-i\omega_\beta\left(\left[F_N(J_\mu)\,,\,F_N(J_\nu)\right]\right)=
\varepsilon_{\mu\nu\delta}\,\omega_\beta\left(J^{(k)}_\delta\right)$,
and the symplectic matrix in \eqref{qfa0a} is given by 
$\displaystyle
\sigma_\beta=\omega_3\begin{pmatrix}0&1&0\cr
-1&0&0\cr
0&0&0\end{pmatrix}$.

Thus, in the large $N$ limit, the emerging fluctuation operators $F(J_{1,2})$ commute with $F(J_3)$, while 
\begin{equation}
\label{symspin1}
\left[F(J_1)\,,\,F(J_2)\right]=i\,\omega_3\ .
\end{equation}
Therefore, in the following we shall concentrate on the fluctuation operators 
$F(J_{1,2})$ as the third one behaves as a decoupled degree of freedom.

Concerning the covariance matrix $\Sigma_\beta$ whose entries are given 
in \eqref{qfa1aa}, because of \eqref{Gibbsspin1}, one gets
$\omega_\beta\left(\left\{F_N(J_\mu)\,,\,F_N(J_\nu)\right\}\right)=
\omega_\beta\left(\left\{J_\mu\,,\,J_\nu\right\}\right)$
whenever at least one index is different from $3$, otherwise
$\omega_\beta\left(\left\{F_3(J_\mu)\,,\,F_N(J_3)\right\}\right)=2\left(\omega_\beta(J^2_3)-\omega_3^2\right)$.
Then, by means of the relations \eqref{spin1+-}, the covariance matrix of $F(J_{1,2})$ reads
\begin{equation}
\label{covspin1}
\Sigma_\beta=\frac{1+\cosh(\beta\omega)}{1+2\cosh(\beta\omega)}\,\bold{1}\ ,
\end{equation}
where $\bold{1}$ is the $2\times 2$ identity matrix.

The quantum fluctuation algebra relative to the selected set of observables is thus generated by the large $N$ limit of operators $W_N(r)=\exp\left(i(r_1F_N(J_1)+r_2F_N(J_2)\right)$. These tend to Weyl operators $W(r)=\exp\left(i(r_1F(J_1)+r_2F(J_2)\right)$ in such a way that
\begin{equation}
\label{messtatespin1}
\lim_{N\toì\infty}\omega_\beta\left(W_N(r)\right)=\exp\left(-\frac{r_1^2+r_2^2}{2}\frac{1+\cosh(\beta\omega)}{1+2\cosh(\beta\omega)}\right)=\Omega_\beta(W(r))\ ,
\end{equation}
where $\Omega_\beta$ is the emergent mesoscopic state on the fluctuation algebra.

We shall now focus upon the microscopic dynamics of the chain that we assume to be  dissipative and generated by local Lindblad generators of the form:
\begin{equation}
\mathbb{L}_N[X]=i\omega\sum_{k=-N}^N\left[\,J_3^{(k)}\,,\,X\right]+\frac{\lambda}{2}\sum_{k=-N}^N\left[\left[J_3^{(k)},X\right],J_3^{(k)}\right]\ ,\qquad \lambda>0\ ,
\end{equation}
for all $X\in\mathcal{A}_{[-N,N]}$. 
As the Kossakowski matrix is $\lambda$ times the $(2N+1)\times(2N+1)$ identity matrix, 
the generated maps $\Phi^N_t=\exp(t\mathbb{L}_N)$ are completely positive and their dual maps preserve the state $\omega_\beta$: $\omega_\beta\circ\Phi^N_t=\omega_\beta$.
By explicitly computing $\mathbb{L}_N[J_{1,2}^{(k)}]$ one finds that the locality conditions of Section \ref{loccondsec} are satisfied and that the matrix $\mathcal{L}$ in \eqref{assump0} equals
$$
\mathcal{L}=-\frac{\lambda}{2}\,\bold{1}-\omega\begin{pmatrix}0&1\cr-1&0\end{pmatrix}\ .
$$
Therefore, the time evolution of the Weyl operators 
$$
W(r)\mapsto W_t(r)={\rm e}^{-(r,\mathcal{Y}_t\,r)}\,W(r_t)
$$ 
in \eqref{final} is straightforwardly computed with $r_t$ and $Y_t$ in \eqref{intmap1} and \eqref{intmap2} given by
$$
r_t={\rm e}^{-\lambda t/2}\begin{pmatrix}
r_1\cos(\omega t)+r_2\sin(\omega t)\cr
-r_1\sin(\omega t)+r_2\cos(\omega t)
\end{pmatrix}\ ,\quad \mathcal{Y}_t=\frac{1+\cosh(\beta\omega)}{2\left(1+2\cosh(\beta\omega)\right)}\left(1-{\rm e}^{-\lambda t}\right)\,\bold{1}\ .
$$

From the commutation relations \eqref{symspin1} one constructs annihilation and creation operators $a,a^\dag$ such that $[a\,,\,a^\dag]=1$ and
\begin{equation}
\label{aF}
F(J_1)=\eta\,\left(a+a^\dag\right)\ ,\quad
F(J_2)=i\,\eta\,\left(a-a^\dag\right)\ ,\quad \eta=\sqrt{\frac{\sinh(\beta\omega)}{1+2\cosh(\beta\omega)}}\ .
\end{equation}
Then,
\begin{equation}
\label{Fa}
a=\frac{F(J_-)}{2\eta}\ ,\qquad
a^\dag=\frac{F(J_+)}{2\eta}\ ,
\end{equation}
so that the creation operator $a^\dag$ emerges in the large $N$ limit from the 
the fluctuation
$$
F_N(J_+)=\frac{1}{\sqrt{2N+1}}\sum_{k=-N}^NJ_+^{(k)}\ ,
$$
scaled by the temperature dependent factor $1/\eta$.

At a first glance, the previous expression looks similar to the bosonization procedure of 
Holstein and Primakoff \cite{HP}.
There, the raising and lowering operators $S^{(j)}_{\pm}$ acting on a ($2S+1$)-dimensional Hilbert space can be represented by bosonic-like operators
$$
S_+=a^\dag\,\sqrt{2S-a^\dag a}\ ,\quad
S_-=\sqrt{2S-a^\dag a}\,a\ ,\quad S_3=a^\dag a-S\ ,
$$
such that $[a\,,\,a^\dag]=1$,  with vacuum state $\vert {\rm vac}\rangle=\vert S,-S\rangle$ so that
$$
a\vert{\rm vac}\rangle
=\frac{1}{\sqrt{S-S_3}}\,S_-\vert S,-S\rangle=0\ .
$$
In the Holstein-Primakoff approximation, one considers $S\gg 1$ and restricts the action of the various operators to states $\vert\psi\rangle$ such that 
$\langle\psi\vert a^\dag a\vert\psi\rangle\ll S$. Then, the creation and annihilation operators arise, in the large $S$ limit,
from spin operators with a fluctuation-like scaling 
\begin{equation}
\label{HP}
a^\dag\simeq\frac{S_+}{\sqrt{2S}}\ ,\qquad a\simeq\frac{S_-}{\sqrt{2S}}\ .
\end{equation}
By using this approximation, one can then proceed to the diagonalisation, for instance of the Heisenberg Hamiltonian, and to the introduction of the notion of bosonic spin-waves or magnons.

Identifying $S_\pm$ with $\sum_{k=-N}^NJ^{(k)}_\pm$ and taking as vacuum state the
tensor product of $2N+1$ spin states all pointing down along the third axis,
\begin{equation}
\label{HPstate}
\vert{\rm vac}\rangle=\bigotimes_{k=-N}^N\vert -1\rangle_k\ ,
\end{equation} 
a similarity emerges between the Holstein-Primakoff approximation and the quantum spin fluctuations which has indeed been applied to the study of magnons and their properties \cite{Michoel}. However, the two approaches are physically quite different.

In general, the Holstein-Primakoff approximation applied to a spin-chain involves the bosonization of the on-site internal spin degree of freedom through the scaling $1/\sqrt{S}$, at each fixed lattice site $-N\leq j\leq N$.
Instead, the fluctuations considered here give rise to the bosonization
of collective spin obervables through the scaling $1/\sqrt{2N+1}$ with respect to the number of lattice sites $2N+1$, with fixed on-site spin dimension. 

In the example studied above, a comparison with the Holstein-Primakoff approximation can only be performed through operators involving fluctuations like $F_N(J_\pm)$, since the on-site spins have fixed dimension.
However, also at this level of comparison important differences emerge. In fact, in the Holstein-Primakoff approximation these fluctuation are directly identified with Bosonic operator as in \eqref{HP}, without the temperature dependent scaling as in \eqref{Fa}.
Indeed, the Holstein-Primakoff approximation is only valid on a sector of the Hilbert space where the mean value of $a^\dag a$ is much smaller than $S=2N+1$, while the fluctuation theory holds for general clustering states. 

In the specific case examined in this section, the mesoscopic state $\Omega_\beta$ on in \eqref{messtatespin1} corresponds to the thermal state 
$$
R_\beta=\left(1-{\rm e}^{-\beta\omega}\right)\,{\rm e}^{-\beta\omega\,a^\dag\,a}\ .
$$
Indeed, using \eqref{aF} and turning the Weil operator $W(r)$ into the displacement operator
$\exp(z\,a^\dag-z^*\,a)$ with $z=i\eta(r_1+ir_2)$, one retrieves
$$
{\rm Tr}\left(R_\beta\,{\rm e}^{z\,a^\dag-z^*\,a}\right)=
\exp\left(-\frac{|z|^2}{2}\coth\frac{\beta\omega}{2}\right)=\Omega_\beta(W(r))\ .
$$
The higher the temperature, the smaller $\beta$ and the stronger are the differences between the two approaches, while they disappear when 
$\beta\to\infty$ and the microscopic state $\omega_\beta$ tend to a tensor product as in \eqref{HPstate}.

\section{Conclusions}

We have considered an open quantum spin chain equipped with a translation invariant, clustering state and subjected to a dissipative dynamics of Lindblad type that acts locally, that is on the spin algebra supported by the sites from $-N$ to $N$.
We have studied the mesoscopic dynamics emerging at the level of the quantum fluctuations of single site spin observables when $N\to+\infty$.
Unlike mean-field observables, in the large $N$ limit, the local fluctuations tend to non-commuting extended observables that obey the canonical commutation relations and generate a Weyl algebra. Morevoer, in the same limit, the microscopic state on the quantum spin chain tends, via a quantum central limit theorem to a Gaussian mesoscopic state on such  an algebra.
Differently from previous works \cite{BCF}, where the microscopic state was chosen to be not only translation-invariant and clustering, but also of factorized form and satisfying the \textit{KMS} thermal conditions, we have here relaxed the latter two constraints. Furthermore, the Lindblad generator presently considered has only been chosen to respect the structure of the linear span of the set of observables that give rise to the global quantum fluctuations, while in \cite{BCF} a specific generator was chosen.
Under these assumptions, we have proved that the emerging dissipative mesoscopic dynamics
on the Weyl algebra of quantum fluctuations generically consists of a semigroup of unital, completely positive maps that transform mesoscopic Gaussian states into states of the same kind. Finally, in a specific example we have emphasised the differences between the theory of quantum spin fluctuations and the Holstein-Primakoff spin bosonization procedure.

\section{Appendix A}
\label{AppA}

\noindent
\textbf{Proof of Lemma \ref{lemma5}:}\quad
Due to the properties of commutators, the lemma is certainly true
for the Hamiltonian term $\mathbb{H}_N$ of the Lindblad generator.
For the dissipative term $\mathbb{D}_N$ we proceed by induction: what the lemma states holds for single-site operators, then, by assuming it to be true for the tensor product of $n$ single-site operators, we show the same to hold for $n+1$ products.
Using the algebraic relation
$$
b\,\Big(a\,[d\,,\,c]\,+\,[a\,,\,d]\,c\Big)\,+\,\Big(a\,[b\,,\,c]\,+\,[a\,,\,b]\,c\Big)\,d\,-\,a\,[bd\,,\,c]\,-\,[a\,,\,bd]\,c\,=\,
-2\,[a\,,\,b]\,[d\,,\,c]\ ,
$$
one derives that 
$$
\mathbb{D}_N[ab]=\mathbb{D}_N[a]\,b\,+\,a\,\mathbb{D}_N[b]\,+\,2\sum_{\mu,\nu=1}^m\,\frac{D_{\mu\nu}}{2}\sum_{k,\ell=-N}^{N}J_{k\ell}\left[v_\mu^{(k)}\,,\,a\right]\,\left[b\,,\,v_\nu^{(\ell)}\right]\ ,
$$
for all operators $a,b\in\mathcal{A}_{[-N,N]}$. Let $a=o_{\alpha_1}^{(k_1)}o_{\alpha_2}^{(k_2)}\dots o_{\alpha_n}^{(k_n)}$ and $b=o_{\alpha_{n+1}}^{(k_{n+1})}$; then,
\begin{eqnarray*}
&&
\mathbb{D}_N\left[o_{\alpha_1}^{(k_1)}o_{\alpha_2}^{(k_2)}\dots o_{\alpha_n}^{(k_n)}o_{\alpha_{n+1}}^{(k_{n+1})}\right]=o_{\alpha_1}^{(k_1)}o_{\alpha_2}^{(k_2)}\dots o_{\alpha_n}^{(k_n)}\,\mathbb{D}_N\left[
o_{\alpha_{n+1}}^{(k_{n+1})}\right]\,+\\
&&\hskip 3cm
+\,\mathbb{D}_N\left[o_{\alpha_1}^{(k_1)}o_{\alpha_2}^{(k_2)}\dots o_{\alpha_n}^{(k_n)}\right]\,o_{\alpha_{n+1}}^{(k_{n+1})}\,+\\
&&\hskip 3cm
+\,2\sum_{\mu\nu=1}^m\,\frac{D_{\mu\nu}}{2}\sum_{k,\ell=-N}^{N}J_{k\ell}\,\left[v_\mu^{(k)}\,,\,o_{\alpha_1}^{(k_1)}o_{\alpha_2}^{(k_2)}\dots o_{\alpha_n}^{(k_n)}\right]\,\left[o_{\alpha_{n+1}}^{(k_{n+1})}\,,\,v_\nu^{\dagger(\ell)}\right]\ .
\label{n+1}
\end{eqnarray*}
Due to the assumptions and to the induction hypothesis, the first two contributions 
can again be expressed as linear combinations of products of single-site basis operators at sites $k_1,k_2,\dots k_{n+1}$.
As for the last term, note that
\begin{eqnarray*}
&&
\sum_{k,\ell=-N}^{N}\,J_{k\ell}\,\left[v_\mu^{(k)}\,,\,o_{\alpha_1}^{(k_1)}o_{\alpha_2}^{(k_2)}\dots o_{\alpha_n}^{(k_n)}\right]\,\left[o_{\alpha_{n+1}}^{(k_{n+1})}\,,\,v_\nu^{(\dagger\ell)}\right]=\\
&&\hskip 2cm
=\sum_{q=k_1,k_2,\dots k_n}\,J_{q k_{n+1}}\,o_{\alpha_1}^{k_1}o_{\alpha_2}^{k_2}\dots o_{\alpha_{k-1}}^{k_{q-1}}\left[v_\mu^{(q)}\,,\,o_{\alpha_q}^{(q)}\right]\,o_{\alpha_{q+1}}^{k_{q+1}}\dots o_{\alpha_n}^{k_{n}}\,\left[o_{\alpha_{n+1}}^{k_{n+1}}\,,\,v_\nu^{\dagger(k_{n+1})}\right]\ .
\end{eqnarray*}
Therefore, by expanding the various commutators with respect to the single-site matrix basis $\{o_\alpha\}_{\alpha=1}^{p^2}$, it can also be written as a linear combination of products of basis operators at sites $k_1,k_2,\dots k_{n+1}$. 
\hfill
\qed

\section{Appendix B}

\noindent
\textbf{Proof of Lemma \ref{lemma1}:}\quad
Using the notation in \eqref{qfa2b} and the unitarity and factorisation of 
$W_N(r)$, 
\begin{eqnarray}
\label{proof1a}
\left[z^{(k)}\,,\,W_N(r)\right]&=&\left(z^{(k)}\,-\,{\rm e}^{i\frac{q_r^{(k)}}{\sqrt{N_T}}}\,z^{(k)}\,{\rm e}^{-i\frac{q_r^{(k)}}{\sqrt{N_T}}}\right)\,W_N(r)=-U_{q^{(k)}_r}[z^{(k)}]\,W_N(r)
\\
\label{proof1b}
U_{q^{(k)}_r}[z^{(k)}]&=&\sum_{n=1}^{\infty}\frac{i^n}{n!\left(\sqrt{N_T}\right)^n}\,\mathbb{K}_{q^{(k)}_r}^n\left[z^{(k)}\right]\ ,
\end{eqnarray}
with $z$ any single site operator and $\mathbb{K}^n_{q_r}[z]$ the multi-commutator defined by
\begin{equation}
\label{multicomm}
\mathbb{K}^n_{q_r}\left[z\right]=\left[q_r,\mathbb{K}^{n-1}_{q_r}[z]\right]\ ,\quad
\mathbb{K}^0_{q_r}[z]=z\ .
\end{equation}
Notice that $\displaystyle U^\dag_{q^{(k)}_r}[z^{(k)}]=U_{q^{(k)}_r}[(z^\dag)^{(k)}]$.

Consider the commutator with the Hamiltonian in \eqref{tih}:
$$
i\left[H_N\,,\,W_N(r)\right]=-i\sum_{k=-N}^{N}\,U_{q_r^{(k)}}[h^{(k)}]\,W_N(r)\ .
$$
In the series expansion of $U_{q_r^{(k)}}[h^{(k)}]$, the relevant contribution is given by the first two terms
$$
\widetilde{U}_{q^{(k)}_r}[h^{(k)}]\,=\,\frac{i}{\sqrt{N_T}}\Big[q^{(k)}_r\,,\,h^{(k)}\Big]
\,-\,\frac{1}{2N_T}\,\Big[q^{(k)}_r\,,\,\Big[q^{(k)}_r\,,\,h^{(k)}\Big]\ .
$$
Indeed, the remaining infinite series vanishes in norm when $N\to+\infty$ as 
$$
\left\|\sum_{n=3}^\infty \frac{i^n}{n!\left(\sqrt{N_T}\right)^n}\,\mathbb{K}_{q^{(k)}_r}^n\left[h^{(k)}\right]\right\|\leq \|h\|\,\sum_{n=3}^\infty \frac{2^n\|q_r\|^n}{n!\left(\sqrt{N_T}\right)^n}\leq
\frac{1}{N^{3/2}_T}\,{\rm e}^{2\|q_r\|}\|h\|\ ,
$$
so that, in the large $N$ limit, the quantity $\sum_{k=-N}^{N}U_{q^{(k)}_r}[h^{(k)}]$ behaves as 
\begin{equation}
\label{proof1c}
\sum_{k=-N}^{N}\widetilde{U}_{q^{(k)}_r}[h^{(k)}]=\sum_{k=-N}^{N}\Bigg(\frac{i}{\sqrt{N_T}}\Big[q^{(k)}_r\,,\,h^{(k)}\Big]
\,-\,\frac{1}{2N_T}\,\Big[q^{(k)}_r\,,\,\Big[q^{(k)}_r\,,\,h^{(k)}\Big]\Bigg)\ .
\end{equation}
for  
$$
\left\|\sum_{k=-N}^{N}\Big(U_{q^{(k)}_r}[h^{(k)}]\,-\,\widetilde{U}_{q^{(k)}_r}[h^{(k)}]\Big)\right\|\leq \frac{1}{N^{1/2}_T}\,{\rm e}^{2\|q_r\|}\|h\|
$$ 
which vanishes when  $N\to+\infty$. 

Using \eqref{qfa2b} and \eqref{assump1}, the first term contributing to \eqref{proof1c} scales as a fluctuation. Since operators at different sites commute, it can be rewritten 
as
\begin{eqnarray}
\label{proof1d}
\frac{i}{\sqrt{N_T}}\sum_{k=-N}^{N}\Big[q^{(k)}_r\,,\,h^{(k)}\Big]
&=&
-\frac{i}{\sqrt{N_T}}\sum_{k,\ell=-N}^{N}\Big[h^{(k)}\,,\,q^{(\ell)}_r\Big]=
-\,i\Big[H_N\,,\,(r,F_N)\Big]\\
\label{proof1f}
&=&-(r,\mathcal{H}F_N)\,-\,\sqrt{N_T}\,(r,\mathcal{H}x_\omega)\ ,
\end{eqnarray}
where $\mathcal{H}$ is the matrix with entries $\mathcal{H}_{ij}$ and $x_\omega\in\mathbb{R}^d$ has components $\omega(x_i)$.
 
Instead, the second term in \eqref{proof1c} scales as a mean field observable; since operators at different sites commute, it can be rearranged as follows: 
\begin{equation}
\sum_{k,\ell,j=-N}^{N}\frac{1}{2N_T}\left[q_r^{(j)},\left[q_r^{(\ell)},h^{(k)}\right]\right]=\frac{i}{2}\left[\left(r,F_N\right),\left(r,\mathcal{H}F_N\right)\right]\ .
\end{equation}
Unlike in the first term, because of the commutators, the scalar 
term $\sqrt{N_T}\,(r,\mathcal{H}x_\omega)$ in \eqref{proof1f} does not contribute and
the second term can be written in terms of the fluctuation vector $F_N=(F_N(x_1),\ldots,F_N(x_d))^{tr}$, only.
\hfill\qed

\section{Appendix C}

\noindent
\textbf{Proof of Lemma \ref{lemma2}:}\quad
The same strategy as in the proof of Lemma \ref{lemma1} now applied to the dissipative contribution to the Lindblad generator first yields
\begin{eqnarray*}
\mathbb{D}_N\left[W_N(r)\right]&=&\sum_{k,\ell=-N}^{N}J_{k\ell}\sum_{\mu,\nu=1}^m\frac{D_{\mu\nu}}{2}\Big(v_\mu^{(k)}\,U_{q^{(k)}_r}[(v_\nu^\dag)^{(\ell)}]\,-\,U_{q^{(k)}_r}[v_\mu^{(k)}]\,(v_\nu^\dagger)^{(\ell)}\\
&-&U_{q^{(k)}_r}[v_\mu^{(k)}]\,
U_{q^{(\ell)}_r}[(v_\nu^\dag)^{(\ell)}]\Big)\,W_N(r)\ .
\end{eqnarray*}
Then, by considering the expansions of the two terms in the last contribution, one shows that, apart from the first summands in each series, the rest can be estimated in norm by:
$$
\left\|\sum_{k,\ell=-N}^{N}J_{k\ell}\sum_{n+m>2}\frac{i^{n+m}}{n!m!}\frac{\mathbb{K}_{q^{(k)}_r}^{n}[v^{(k)}_\mu]\,\mathbb{K}_{q^{(\ell)}_r}^{m}[(v_\nu^\dagger)^{(\ell)}]}{\sqrt{N^{(m+n)}_T}}\right\|\leq
\frac{1}{N^{3/2}_T}\sum_{k,\ell=-N}^{N}\left|J_{k\ell}\right|e^{4\|q_r\|}\|v_\mu\|\|v_\nu\|\ .
$$
Because of the assumption \eqref{assum0} on the coefficients $J_{k\ell}$, it then follows that
$$
\lim_{N\to+\infty}\left\|\sum_{k,\ell=-N}^{N}J_{k\ell}\sum_{\mu,\nu=1}^m\,\frac{D_{\mu\nu}}{2N_T}\left(i\left[q_r^{(k)},v_{\mu}^{(k)}\right]\,i\left[q_r^{(\ell)},(v_{\nu}^\dagger)^{(\ell)}\right]-U_{q^{(k)}_r}[v_\mu^{(k)}]\,U_{q^{(\ell)}_r}[(v_\nu^\dag)^{(\ell)}]\right)\right\|=0\ .
$$
Using similar arguments as before, one can also show that the other two contributions
essentially amount to the first two terms in the series expansion; indeed,
\begin{eqnarray*}
\lim_{N\to+\infty}\Bigg\|\sum_{k,\ell=-N}^{N}J_{k\ell}\sum_{\mu,\nu=1}^m\,\frac{D_{\mu\nu}}{2}\Bigg(\frac{1}{\sqrt{N_T}}v_{\mu}^{(k)}\,i\left[q_r^{(\ell)}\,,\,(v_{\nu}^\dagger)^{(\ell)}\right]&-&\frac{1}{2N_T}v_{\mu}^{(k)}\,\left[q_r^{(\ell)}\,,\,\left[q_r^{(\ell)}\,,\,(v_{\nu}^\dagger)^{(\ell)}\right]\right]\\
&-&v_\mu^{(k)}\,
U_{q^{(\ell)}_r}[(v_\nu^\dag)^{(\ell)}]\Bigg)\Bigg\|=0\\
\lim_{N\to+\infty}\Bigg\|\sum_{k,\ell=-N}^{N}J_{k\ell}\sum_{\mu,\nu=1}^m\,\frac{D_{\mu\nu}}{2}\Bigg(\frac{1}{\sqrt{N_T}}\,i\left[q_r^{(k)},v_{\mu}^{(k)}\right]\,(v_{\nu}^\dagger)^{(\ell)}&-&\frac{1}{2N_T}\left[q_r^{(k)}\,,\,\left[q_r^{(k)}\,,\,v_{\mu}^{(k)}\right]\right]\,(v_{\nu}^\dagger)^{(\ell)}\\
&-&U_{q^{(k)}_r}[v_\mu^{(k)}]\,(v_{\nu}^\dagger)^{(\ell)}\Bigg)\Bigg\|=0\ .
\end{eqnarray*}
Thus, for large $N$, the action of the dissipative part of $\mathbb{L}_N$ can be approximated by
\begin{eqnarray}
\label{proof1e1}
&&\hskip-1.5cm
\mathbb{D}_N\left[W_N(r)\right]\simeq
\sum_{k,\ell=-N}^{N}J_{k\ell}\sum_{\mu,\nu=1}^m\,\frac{D_{\mu\nu}}{2\sqrt{N_T}}\left(i\left[v_\mu^{(k)}\,,\,q_r^{(k)}\right]\,v_\nu^{\dagger(\ell)}\,+\,i\,v_\mu^{(k)}\left[q_r^{(\ell)}\,,\,(v_\nu^{\dagger})^{(\ell)}\right]\right)\,W_N(r)\\
\label{proof1e2}
&&\hskip .5cm+\sum_{k,\ell=-N}^{N}J_{k\ell}\sum_{\mu,\nu=1}^m\,\frac{D_{\mu\nu}}{2N_T}\Bigg(\left[q_r^{(k)},v_\mu^{(k)}\right]\,\left[q_r^{(\ell)},(v_\nu^\dagger)^{(\ell)}\right]
\\
\label{proof1e3}
&&\hskip .5cm
-\frac{1}{2}v_\mu^{(k)}\,\left[q_r^{(\ell)},\left[q_r^{(\ell)}\,,\,(v_{\nu}^\dagger)^{(\ell)}\right]\right]\,+\,\frac{1}{2}\left[q_r^{(k)},\left[q_r^{(k)}\,,\,v_{\mu}^{(k)}\right]\right]\,(v_{\nu}^\dagger)^{(\ell)}\Bigg)\,W_N(r)\ .
\end{eqnarray}
Using \eqref{qfa2b} and \eqref{assump1}, the term \eqref{proof1e1} that scales as $1/\sqrt{N_T}$ can be written as:
\begin{eqnarray}
\label{proof1e4}
&&\hskip-1cm
\frac{i}{2\sqrt{N_T}}\sum_{k,\ell=-N}^{N}J_{k\ell}\sum_{\mu,\nu=1}^m\,D_{\mu\nu}\,\Big(\left[v_\mu^{(k)}\,,\,q_r^{(k)}\right]\,(v_\nu^\dagger)^{(\ell)}\,+\,v_\mu^{(k)}\left[q_r^{(\ell)}\,,\,(v_\nu^\dagger)^{(\ell)}\right]\Big)=\\
\label{proof1e5}
&&\hskip 2cm
=i\,\mathbb{D}[(r,F_N)]=i\,(r,\mathcal{D}F_N)\,+\,i\sqrt{N_T}\,(r,\mathcal{D}x_\omega)\ ,
\end{eqnarray}
where as in \eqref{proof1f} $x_\omega\in\mathbb{R}^d$ is a real vector with components $\omega(x_i)$.

Concerning the term in \eqref{proof1e3}, by using \eqref{qfa2b} and the fact that operators at different sites commute, it can be recast in the form
\begin{eqnarray*}
p^{(k,\ell)}_{{\mu\nu}}(r,N)&:=&-\frac{1}{2N_T}v_\mu^{(k)}\left[q_r^{(\ell)}\,,\,\left[q_r^{(\ell)}\,,\,(v_{\nu}^\dagger)^{(\ell)}\right]\right]\,+\,\frac{1}{2N_T}\left[q_r^{(k)}\,,\,\left[q_r^{(k)}\,,\,v_{\mu}^{(k)}\right]\right]\,(v_{\nu}^\dagger)^{(\ell)}\\
&=&-\frac{1}{2}\,v_\mu^{(k)}\left[(r,F_N)\,,\,\left[(r,F_N)\,,\,(v_\nu^\dagger)^{(\ell)}\right]\right]\,+\,\frac{1}{2}\left[(r,F_N)\,,\,\left[(r,F_N),v_{\mu}^{(k)}\right]\right]\,(v_{\nu}^\dagger)^{(\ell)} \ .
\end{eqnarray*}
Since
$
-a\,\big[b\,,\,\big[b\,,\,c\big]\big]\,+\,\big[b\,,\,\big[b\,,\,a\big]\big]\,c\,=\,\Big[b\,,\Big(\,a\,[b\,,\,c]\,+\,[a\,,\,b]\,c\Big)\Big]$,
one can finally write
$$
\sum_{k,\ell=-N}^{N}J_{k\ell}\sum_{\mu,\nu=1}^m\,\frac{D_{\mu\nu}}{2}p^{(k,\ell)}_{{\mu\nu}}(r;N)=-\frac{1}{2}\left[(r,F_N)\,,\,\mathbb{D}_N\left[\left(r,F_N\right)\right]\right]\ .
$$
Further, since $\mathbb{D}_N\left[\left(r,F_N\right)\right]$ appears inside a commutator, the scalar quantity in \eqref{proof1e5} does not contribute, whence 
$$
\sum_{k,\ell=-N}^{N}J_{k\ell}\sum_{\mu,\nu=1}^m\,\frac{D_{\mu\nu}}{2}p^{(k,\ell)}_{{\mu\nu}}(r;N)=-\frac{1}{2}\left[(r,F_N)\,,\,(r,\mathcal{D}F_N)\right]\ .
$$

Let us now consider the contribution in \eqref{proof1e2} that scales as $1/N_T$. 
A similar argument as before allows us to recast it as 
$$
s_{\mu\nu}^{(k,\ell)}(r;N)=\frac{1}{N_T}\left[q_r^{(k)},v_\mu^{(k)}\right]\left[q_r^{(\ell)}\,,\,(v_\nu^\dagger)^{(\ell)}\right]
=\left[(r,F_N)\,,\,v_\mu^{(k)}\right]\,\left[(r,F_N)\,,\,(v_\nu^\dagger)^{(\ell)}\right]\ .
$$
Using the algebraic relation
$$
b\,\Big(a\,[d\,,\,c]\,+\,[a\,,\,d]\,c\Big)\,+\,\Big(a\,[b\,,\,c]\,+\,[a\,,\,b]\,c\Big)\,d\,-\,a\,[bd\,,\,c]\,-\,[a\,,\,bd]\,c\,=\,
2\,[b\,,\,a]\,[d\,,\,c]\ ,
$$
we get:
\begin{eqnarray*}
S(r;N)&:=&\sum_{k,\ell=-N}^{N}J_{k\ell}\sum_{\mu,\nu=1}^m\,\frac{D_{\mu\nu}}{2}s_{\mu\nu}^{(k,\ell)}(r;N)\\
&=&\frac{1}{2}\Big(\mathbb{D}_N\left[(r,F_N)\right]\,(r,F_N)\,+\,(r,F_N)\mathbb{D}_N\left[(r,F_N)\right]\,-\,\mathbb{D}_N\left[(r,F_N)^2\right]\Big)\ .
\end{eqnarray*}
Moreover, since the Hamiltonian term of the Lindblad generator is such that
$$
\mathbb{H}_N\left[(r,F_N)^2\right]=\mathbb{H}_N\left[(r,F_N)\right]\,(r,F_N)+(r,F_N)\,\mathbb{H}_N\left[(r,F_N)\right]\ ,
$$ 
one can write $S(r;N)$ using the full Lindblad generator:
$$
S(r;N)=\frac{1}{2}\Big(\mathbb{L}_N\left[(r,F_N)\right]\,(r,F_N)\,+\,(r,F_N)\,\mathbb{L}_N\left[(r,F_N)\right]\,-\,\mathbb{L}_N\left[(r,F_N)^2\right]\Big)\ .
$$
\hfill\qed

\section{Appendix D}

\noindent
\textbf{Proof of Lemma \ref{lemma4}:}\quad
From the algebraic relation $\displaystyle 
\left[{\rm e}^{iA}\,,\,B\right]=\int_0^1dy\frac{d}{dy}\left({\rm e}^{iyA}\,B\,{\rm e}^{i(1-y)A}\right)$ it follows that
$\displaystyle
\left\|\left[W_N(r)\,,\,R_N\right]\right\|\le\left\|\left[(r,F_N),R_N\right]\right\|$.
From
\begin{equation}
\left[(r,F_N),R_N\right]=\frac{1}{N_T^{3/2}}\sum_{i=1}^d\,r_i\sum_{k,\ell=-N}^{N}J_{k\ell}\left(\left[x_i^{(k)}\,,\,a^{(k)}\right]\,b^{(\ell)}\,+\,a^{(k)}\,\left[x_i^{(\ell)}\,,\,b^{(\ell)}\right]\right)\ ,
\end{equation}
the upper bound 
$$
\left\|\left[(r,F_N),R_N\right]\right\|\leq 4d\max_{1\leq i\leq d}\left\{|r_i|\|x_i\|\right\}\,\|a\|\|b\|\,\frac{1}{N_T^{3/2}}\sum_{k,\ell=-N}^N|J_{k\ell}|
$$
follows. It vanishes when $N\to+\infty$; indeed, the hypothesis on the coefficients $J_{k\ell}$ yields
\begin{eqnarray*}
\nonumber
\lim_{N\to+\infty}\frac{1}{N_T^{3/2}}\sum_{k,\ell=-N}^N|J_{k\ell}|&=&\lim_{N\to+\infty}\frac{1}{N_T^{3/2}}\sum_{k=-N}^N\sum_{p=k-N}^{k+N}|J(p)|
\leq\lim_{N\to+\infty}\frac{1}{N_T^{3/2}}\sum_{k=-N}^N\sum_{p=-\infty}^{+\infty}|J(p)|\\
&\leq&\lim_{N\to+\infty}
\frac{1}{\sqrt{N_T}}\sum_{p=-\infty}^{+\infty}|J(p)|=0\ .
\end{eqnarray*}
This proves the first result of the lemma. The second one amounts to showing that 
$$
\lim_{N\to+\infty}\omega\left(R_N^\dagger R_N\right)=\left|R\right|^2\ ,
$$ 
where
$$
\omega\left(R_N^\dagger R_N\right)=\frac{1}{N_T^2}\sum_{k,\ell=-N}^{N}\sum_{n,m=-N}^{N}J^*_{k\ell}\,J_{nm}\omega\left((b^\dagger)^{(\ell)}\,(a^\dagger)^{(k)}\,a^{(n)}\,b^{(m)}\right)
\ .
$$
Using the translation invariance of $\omega$ we write
$$
\omega\left((b^\dagger)^{(\ell)}\,(a^\dagger)^{(k)}\,a^{(n)}\,b^{(m)}\right)=\omega\left(\tau^{(\ell-m)}\left(b^\dagger\,(a^\dagger)^{(k-\ell)}\right)\,a^{(n-m)}\,b\right)\ .
$$
Then, by setting $p=n-m$ and $j=k-l$, we estimate
\begin{eqnarray*}
\Big|\omega\left(R_N^\dagger R_N\right)-\omega\left(R_N^\dagger\right)\omega\left(R_N\right)\Big|&\leq&\frac{1}{N_T^2}\sum_{\ell,m=-N}^{N}\sum_{p=-N-m}^{N-m}\sum_{j=-N-\ell}^{N-\ell}\left|J(j)\right|\,\left|J(p)\right|\times\\
&\times&\Bigg|\omega\left(\tau^{(\ell-m)}\left(b^{\dagger}\,(a^\dagger)^{(j)}\right)\,a^{(p)}\,b\right)-\omega\left(b^\dagger\, (a^\dagger)^{(j)}\right)\,\omega\left(a^{(p)}\,b\right)\Bigg|\\
&\leq&\frac{1}{N_T}\sum_{h=-N_T+1}^{N_T-1}\sum_{p=-\infty}^{\infty}\sum_{j=-\infty}^{\infty}\left|J(j)\right|\,\left|J(p)\right|\times\\
&\times&\Bigg|\omega\left(\tau^{(h)}\left(b^{\dagger}\,(a^\dagger)^{(j)}\right)\,a^{(p)}\,b\right)-\omega\left(b^\dagger\,(a^\dagger)^{(j)}\right)\,\omega\left(a^{(p)}\,b\right)\Bigg|\ .
\end{eqnarray*}
The two infinite sums converge uniformly in the summation index $h$ because 
$$
\Bigg|\omega\left(\tau^{(h)}\left(b^{\dagger}\,(a^\dagger)^{(j)}\right)\,a^{(p)}\,b\right)-\omega\left(b^\dagger\,(a^\dagger)^{(j)}\right)\,\omega\left(a^{(p)}\,b\right)\Bigg|\leq 2\,\|a\|^2\,\|b\|^2\ ,
$$
and because of the assumptions on the coefficients $J_{k\ell}$; therefore \footnote{Notice that, if a sequence $\{z_k\}_k$ is such that $\displaystyle z:=\lim_{k\to+\infty}z_k$ exists then $\displaystyle \lim_{N\to+\infty}\frac{1}{N}\sum_{k=1}^Nz_k=z$.}, 
\begin{eqnarray*}
&&\lim_{N\to+\infty}\left|\omega\left(R_N^\dagger R_N\right)-\omega\left(R_N^\dagger\right)\omega\left(R_N\right)\right|\leq\\
&&\hskip 1cm
\leq\sum_{p=-\infty}^{\infty}\sum_{j=-\infty}^{\infty}\left|J(j)\right|\,\left|J(p)\right|
\,\lim_{h\to\infty}\Bigg|\omega\left(\tau^{(h)}\left(b^{\dagger}\,a^{\dagger(j)}\right)\,a^{(p)}\,b\right)-\omega\left(b^\dagger\, a^{\dagger(j)}\right)\,
\omega\left(a^{(p)}\,b\right)\Bigg|\\
&&\hskip 1cm
+\,\sum_{p=-\infty}^{\infty}\sum_{j=-\infty}^{\infty}\left|J(j)\right|\,\left|J(p)\right|\,\lim_{h\to-\infty}\Bigg|\omega\left(\tau^{(h)}\left(b^{\dagger}\,a^{\dagger(j)}\right)\,a^{(p)}\,b\right)-\omega\left(b^\dagger\,a^{\dagger(j)}\right)\,
\omega\left(a^{(p)}\,b\right)\Bigg|\ .
\end{eqnarray*}
The result follows since the clustering properties of $\omega$ give
$$
\lim_{h\to\pm\infty}\Bigg|\omega\left(\tau^{(h)}\left(b^{\dagger}\,a^{\dagger(j)}\right)\,a^{(p)}\,b\right)-\omega\left(b^\dagger\, a^{\dagger(j)}\right)\,
\omega\left(a^{(p)}\,b\right)\Bigg|=0\ .
$$
\hfill\qed

\section{Appendix E}
\label{AppB}

\noindent
\noindent\textbf{Proof of Lemma \ref{applemma}:}\quad
Given matrices $A$ and $B$, one has
$$
{\rm e}^{iA}\,B\,{\rm e}^{-iA}=\sum_{n=0}^\infty\frac{i^n}{n!}\underbrace{\Big[A\,\Big[A\,,\cdots\Big[}_{n\ times}A\,,\,B\Big]\cdots\Big]\Big]=\sum_{n=0}^\infty\frac{i^n}{n!}\,\mathbb{K}_A[B]\ .
$$
Then, $[N_t\,,\,M_t]=0$ and $N_tN_t^\dag=N_t^\dag N_t=1$ imply $N_tM_tN^\dag_t=M_t$ and
$\displaystyle \dot{N}_t\,N^\dag_t=-N_t\,\dot{N}^\dag_t$. Therefore,
$$
N_t\,\dot{M}_t\,N^\dag_t\,-\,\dot{M}_t\,=-\,\dot{N}_t\,M_t\,N^\dag_t\,-\,N_t\,M_t\,\dot{N}^\dag_t=\Big[M_t\,,\,\dot{N}_t\Big]\,N_t^\dag\ .
$$
Furthermore, since, for $n\geq 1$,
$\displaystyle \mathbb{K}_A^n[B]=\Big[A\,,\,\mathbb{K}^{n-1}_A[B]\Big]$, it follows that
$$
\hskip -.5cm
N_t\,\dot{M}_t\,N^\dag_t-\dot{M}_t=\sum_{n=1}^\infty\frac{i^n}{n!}\mathbb{K}^n_{M_t}[\dot{M}_t]=\Big[M_t\,,\,O_t\Big]=\Big[M_t\,,\,\dot{N}_t\Big]\,N_t^\dag\ ,
$$
where $O_t=\sum_{k=1}^\infty\frac{i^k}{k!}\mathbb{K}_{M_t}^{k-1}[\dot{M}_t]$. 
Then, using again that $[N_t\,,\,M_t]=0$, one obtains
$$
\Big[M_t\,,\,O_t\,N_t\Big]=\Big[M_t\,,\,\dot{N}_t\Big]\ .
$$
In order to show that $\dot{N}_t=O_tN_t$, consider the orthogonal eigenvectors $\vert m_a(t)\rangle$ of $M_t$ with eigenevalues $m_a(t)$. Then, if $m_a(t)\neq m_b(t)$, the previous equality yields 
$$
\langle m_a(t)\vert O_tN_t\vert m_b(t)\rangle=\langle m_a(t)\vert\dot{N}_t\vert m_b(t)\rangle\ .
$$
On the other hand if $\vert m_a(t)\rangle$ and $\vert m_b(t)\rangle$ correspond to a same (real) eigenvalue $m(t)$, then one uses that
$$
0=\frac{{\rm d}}{{\rm d}t}\Big(\langle m_a(t)\vert m_b(t)\rangle\Big)=\langle \dot{m}_a(t)\vert m_b(t)\rangle\,+\,
\langle m_a(t)\vert\dot{m}_b(t)\rangle\ ,
$$
to deduce that also in such a case
\begin{eqnarray*}
\langle m_a(t)\vert O_t\,N_t\vert m_b(t)\rangle&=&i\,\langle m_a(t)\vert \dot{M}_t\vert m_b(t)\rangle\, {\rm e}^{im(t)}\, 
\delta_{ab}
=i\dot{m}(t)\,{\rm e}^{im(t)}\,\delta_{ab}\\ 
&=&\langle m_a(t)\vert\dot{N}_t\vert m_b(t)\rangle\ .
\end{eqnarray*}
\hfill\qed

\end{document}